\documentclass[a4paper]{scrartcl}

\usepackage[utf8]{inputenc}
\usepackage[T1]{fontenc}
\usepackage{lmodern}
\usepackage{tikz}

\usepackage{amsmath}
\usepackage{amssymb}
\usepackage{amsthm}
\usepackage{mathtools} 

\usepackage{hyperref}

\newtheorem{theorem}{Theorem}
\newtheorem{lemma}[theorem]{Lemma}
\newtheorem{corollary}[theorem]{Corollary}

\theoremstyle{definition}
\newtheorem{remark}[theorem]{Remark}

\DeclareMathOperator{\st}{\,s.t.}
\DeclareMathOperator{\argmin}{argmin}

\newcommand{\IEfconn}{\mathcal{I}_{E^f\,\text{conn}}}
\newcommand{\IEfforest}{\mathcal{I}_{E^f\,\text{forest}}}
\newcommand{\IEftree}{\IEfconn\cap\IEfforest}
\newcommand{\IEfmatching}{\mathcal{I}_{E^f\,\text{matching}}}
\newcommand{\IElconn}{\mathcal{I}_{E^\ell\,\text{conn}}}
\newcommand{\IElforest}{\mathcal{I}_{E^\ell\,\text{forest}}}

\newcommand{\IEfE}{\mathcal{I}_{E^f\,\text{all}}}

\begin{document}

\title{On the Complexity of the\\ Bilevel Minimum Spanning Tree
  Problem\thanks{This work has partially been supported by Deutsche
    Forschungsgemeinschaft (DFG) under grant no.~BU 2313/6.}}

\author{Christoph Buchheim, Dorothee Henke, Felix Hommelsheim\thanks{\texttt{\{christoph.buchheim,dorothee.henke,felix.hommelsheim\}@math.tu-dortmund.de}} \\[1ex]
  Department of Mathematics, TU Dortmund University, Germany\\[1ex]}

\date{}

\maketitle

\begin{abstract}
  We consider the bilevel minimum spanning tree (BMST) problem where
  the leader and the follower choose a spanning tree together,
  according to different objective functions. By showing that this
  problem is NP-hard in general, we answer an open question stated
  in~\cite{shi19}. We prove that BMST remains hard even in the special
  case where the follower only controls a matching. Moreover, by
  a polynomial reduction from the vertex-disjoint Steiner
  trees problem, we give some evidence that BMST might even remain hard
  in case the follower controls only few edges.

  On the positive side, we present a polynomial-time
  $(|V|-1)$-approximation algorithm for BMST, where~$|V|$ is the number of
  vertices in the input graph. Moreover, considering the number of
  edges controlled by the follower as parameter, we show that
  2-approximating BMST is fixed-parameter tractable and that, in case
  of uniform costs on leader's edges, even solving BMST exactly is
  fixed-parameter tractable.  We finally consider bottleneck variants
  of BMST and settle the complexity landscape of all combinations of
  sum or bottleneck objective functions for the leader and follower,
  for the optimistic as well as the pessimistic setting.

\end{abstract}

\newpage
\section{Introduction}

A bilevel optimization problem models the interplay between two
decision makers, each of them having their own decision variables,
objective function and constraints. The two decisions can depend on
each other and are made in a hierarchical way: the leader decides
first and the follower second, already knowing what the leader has
decided. However, the problem is usually viewed from the leader's
perspective, who has perfect knowledge of the follower's problem and
takes into account how the follower will react to her\footnote{Throughout this paper, we will refer to the leader using \emph{she/her} and to the follower using \emph{he/him/his}.} decision. In
other words, the optimality of the follower's decision can be viewed
as a constraint in the leader's optimization problem. Several surveys
and text books on bilevel optimization have been
published~\cite{colson07,dempe03,dempe15}.  Bilevel optimization
problems turn out to be very hard in general. Even in the case where
both objective functions and all constraints are linear, they are
strongly NP-hard~\cite{hansen92}. For more details concerning the
complexity of bilevel linear optimization, see~\cite{deng98}.

In this paper, we investigate the complexity of a fundamental
\emph{combinatorial} bilevel optimization problem, namely the bilevel
minimum spanning tree problem. Here, each of the two decision makers
controls a subset of the edges of a given graph and chooses some of
them, such that all chosen edges together form a spanning tree in the
graph.

In a possible application~\cite{shi19}, the two decision makers can be
imagined as a central and a local government whose common task is to
design a transportation network connecting a given set of
facilities. Hence, together they have to construct a spanning tree,
where each of the decision makers may control a different set of
potential links, e.g., federal highways are controlled by the central
government in contrast to local roads, or building grounds are owned
by different actors. First, the central government constructs some of
the connections, and then the local government decides how to complete
the network. The central government, as the leader, pays for the
construction of all connections, while the follower optimizes a
different objective function, e.g., respecting requests of the
citizens.  Another possible application is described
in~\cite{gassner02} and deals with a communication network between
cities. A state is modeled as the leader and assumed to own some
communication connections, while a private company, the follower, is
permitted to build its own connections in specific places. The state
subsidizes the connections built by the company, i.e., the actual
building costs are shared between the two actors. In turn, the company
is required to ensure that each city can communicate with each other
city using some activated state-owned and the new private connections,
i.e., that the result is a spanning tree. The decision process now has
a hierarchical structure: first, the state decides which of its own
connections are to be activated and thus paid for. Second, the company
decides which new connections to build such that all cities are
connected and the costs incurred for the company are
minimized. However, the company's decision influences the state's
costs as well, because of the subsidies. Hence, the state has to
anticipate what the company will do already when making the first
decision.

More formally, let~$G = (V, E)$ be a connected, not necessarily simple
graph with some edges~$E^\ell$ being controlled by the leader and some
edges~$E^f$ being controlled by the follower, such that~$E = E^\ell
\cup E^f$.  Without loss of generality, we assume that~$E^\ell$
and~$E^f$ are disjoint sets, using parallel edges otherwise. We are
given cost functions~$c\colon E \rightarrow \mathbb{R}_{\geq 0}$
and~$d\colon E\rightarrow \mathbb{R}_{\geq 0}$ for the leader and
follower, respectively, and define $c(Z) \coloneqq \sum_{e \in Z}
c(e)$ and~$d(Z) \coloneqq \sum_{e \in Z} d(e)$ for any edge
set~$Z\subseteq E$. With these definitions, the bilevel minimum
spanning tree problem can be formulated as follows:
\begin{equation}\tag{BMST}\label{prob:bmst}
\begin{aligned}
  \min~ & c(X\cup Y)\\
  \st~ & X \subseteq E^\ell \\
  & Y \in
  \begin{aligned}[t]
    \argmin~ & d(Y)\\
    \st~ & Y \subseteq E^f\\
    & X\cup Y \text{ is a spanning tree in }G\;.
  \end{aligned}
\end{aligned}
\end{equation}
Here and in the remainder of the paper, we identify subgraphs of~$G$,
in particular trees and forests, with the corresponding subsets
of~$E$.

If the leader chooses some edge set~$X$ rendering the follower's
problem infeasible, then by definition this choice is not valid for
her. In particular, the leader must choose a cycle-free subset~$X$
of~$E^\ell$ such that the graph~$(V,X\cup E^f)$ is connected, and the
follower will augment~$X$ to a spanning tree at minimum cost according
to his own objective function~$d$. The objective function minimized by
the leader is the total cost of the resulting spanning tree with
respect to the objective function~$c$.

Given a feasible leader's choice~$X$, the follower's problem can
easily be solved in polynomial time, e.g., by Kruskal's
algorithm~\cite{kruskal56} applied to the graph resulting from~$G$ by
contracting all edges in~$X$ and restricting to the edges
in~$E^f$. However, the follower's optimum solution might not be
unique. In order to make the problem well-defined in this case, we
will always assume that the follower chooses his solution greedily
according to some given order of preference that is consistent with
his cost function~$d$. It is easy to verify that both the optimistic
and the pessimistic version of BMST can be modeled in this way. These
are the most common strategies to resolve non-uniqueness of follower's
optimum solutions in bilevel optimization. In the former, the follower
is assumed to decide in favor of the leader among his optimum
solutions, i.e., he uses the leader's objective function as a second
criterion in his optimization. In contrast, the pessimistic view
corresponds to the follower deciding worst possible for the
leader. Note that the follower's feasible set is uniquely determined
by the connected components of the graph~$(V,X)$, and therefore also
his response~$Y$ when assuming any deterministic strategy to resolve
non-uniqueness.

In Fig.~\ref{fig:example}, we give an example of a BMST instance and
its optimum solution. The cost of the leader's optimum solution is
$9$. In contrast to the follower, it is not optimum for the leader to
choose edges in a greedy way since taking the edge~$\{v_4, v_6\}$ into
her solution would result in overall costs of at least $10$. It is
cheaper for her to let the follower connect the
components~$\{v_2,v_3,v_4\}$ and~$\{v_5,v_6\}$ with each
other. However, this strategy relies on the fact that the
edge~$\{v_3,v_6\}$ is cheaper than~$\{v_3,v_5\}$ also for the
follower.
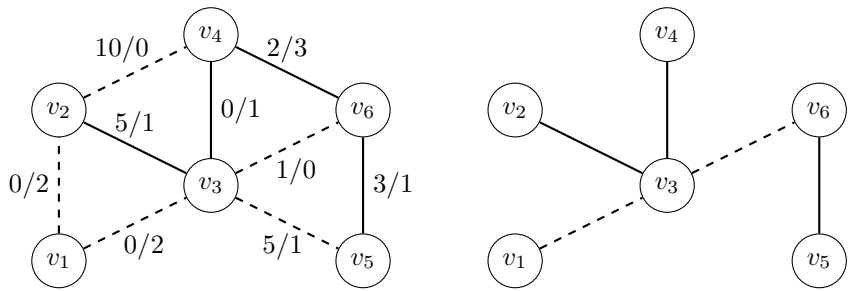
\begin{figure}
  \centering
  \small
  \begin{tikzpicture}[scale=1]
    \node[draw,circle,minimum width = 0.6cm] (v1) at (0,0) {$v_1$};
    \node[draw,circle,minimum width = 0.6cm] (v2) at (0,2) {$v_2$};
    \node[draw,circle,minimum width = 0.6cm] (v3) at (2,1) {$v_3$};
    \node[draw,circle,minimum width = 0.6cm] (v4) at (2,3) {$v_4$};
    \node[draw,circle,minimum width = 0.6cm] (v5) at (4,0) {$v_5$};
    \node[draw,circle,minimum width = 0.6cm] (v6) at (4,2) {$v_6$};

    \draw[thick] (v2) to node[above] {$5 / 1$} (v3);
    \draw[thick] (v3) to node[right] {$0 / 1$} (v4);
    \draw[thick] (v4) to node[above] {$2 / 3$} (v6);
    \draw[thick] (v5) to node[right] {$3 / 1$} (v6);

    \draw[thick,dashed] (v1) to node[left] {$0 / 2$} (v2);
    \draw[thick,dashed] (v1) to node[below] {~~$0 / 2$} (v3);
    \draw[thick,dashed] (v2) to node[above] {$10 / 0$~~~~} (v4);
    \draw[thick,dashed] (v3) to node[below] {$5 / 1$~~} (v5);
    \draw[thick,dashed] (v3) to node[below] {~~$1 / 0$} (v6);

    \node[draw,circle,minimum width = 0.6cm] (u1) at (6,0) {$v_1$};
    \node[draw,circle,minimum width = 0.6cm] (u2) at (6,2) {$v_2$};
    \node[draw,circle,minimum width = 0.6cm] (u3) at (8,1) {$v_3$};
    \node[draw,circle,minimum width = 0.6cm] (u4) at (8,3) {$v_4$};
    \node[draw,circle,minimum width = 0.6cm] (u5) at (10,0) {$v_5$};
    \node[draw,circle,minimum width = 0.6cm] (u6) at (10,2) {$v_6$};

    \draw[thick] (u2) to (u3);
    \draw[thick] (u3) to (u4);
    \draw[thick] (u5) to (u6);

    \draw[thick,dashed] (u1) to (u3);
    \draw[thick,dashed] (u3) to (u6);

  \end{tikzpicture}
  \caption{Example of a BMST instance and its optimum solution
    together with the corresponding response of the follower. The edge
    sets $E^\ell$ and $E^f$ are represented as
    solid and dashed edges, respectively. The labels show the leader's
    and the follower's cost of an edge~$e\in E$ in the form~$c(e) /
    d(e)$.}
  \label{fig:example}
\end{figure}

Besides the problem with sum objective functions, we will also
consider bottleneck versions of BMST, meaning that the leader and/or
the follower only pay for the most expensive edge instead of the sum
over the costs of all chosen edges. In case the follower has a
bottleneck objective function, one has to distinguish between two
possible models: either the follower pays for the most expensive edge
he chooses himself, i.e., his objective function is to
minimize~$\max_{e\in Y} d(e)$, or he pays for the most expensive edge
chosen by any of the two actors, i.e., he minimizes~$\max_{e\in X\cup
  Y} d(e)$; the latter case is the only situation in which the
follower's cost~$d$ of edges in~$E^\ell$ is relevant. These two models
are not equivalent, in contrast to the sum objective case,
where~$d(X\cup Y)=d(X)+d(Y)$ in any optimum solution and hence the two
objectives only differ by~$d(X)$, which is constant from the
follower's perspective.

Under the assumption that the leader's and the follower's edge sets
are not disjoint, but that the follower controls all edges, i.e.,
that~$E^\ell \subseteq E^f=E$, it has recently been shown by Shi
\emph{et al.}~\cite{shi19} that BMST is tractable in case the leader
or the follower (or both) optimize a bottleneck instead of the sum
objective function, where the follower is assumed to
minimize~$\max_{e\in X\cup Y} d(e)$ in the bottleneck case.  Related
results have also been obtained by Gassner~\cite{gassner02}. She
considered the problem version in which $E^\ell$ and $E^f$ are
disjoint and the follower's objective is $\max_{e\in Y} d(e)$ in the
bottleneck case. Polynomial-time algorithms are presented for the
cases where the leader has a bottleneck objective and the follower
either has a sum or a bottleneck objective, while restricting to the
pessimistic problem version in the latter case. In~\cite{shi20},
(single-level) mixed integer linear programming formulations for some
variants of BMST are derived. For exact solution methods for general
bilevel mixed integer programs, we refer to the
survey~\cite{kleinert21}.

Other variants of bilevel optimization problems dealing with minimum
spanning trees are considered in the literature, but in contrast to
the problem addressed here, they usually assume the leader to choose
the prices of some edges, while the follower solves a minimum spanning
tree problem on all edges according to these costs; see~\cite{labbe21}
and the references therein. Also the similar setting in which the
lower level problem is a shortest path problem has been investigated
several times; see the surveys~\cite{vanhoesel08} and~\cite{labbe13}.
Gassner and Klinz~\cite{gassner09} studied a bilevel assignment
problem in which leader and follower choose a perfect matching
together, each of them having their own objective function on the
edges, very similar to the bilevel minimum spanning tree problem
studied here. Sum and bottleneck objective functions are considered,
and it is shown that in most cases, the problem is NP-hard. Only the
optimistic problem version in which both decision makers have
bottleneck objectives remains open.

The authors of~\cite{shi19} conjecture that the version of BMST in
which both leader and follower have a sum objective is NP-hard. Our
main result is a proof of this conjecture. More specifically, we show
that BMST is at least as hard as the Steiner forest problem, hence it
is not approximable to within a factor of $\tfrac{96}{95}$
unless~P$\,=\,$NP. We can show the same result for the special case
where the follower only controls a matching, and give some evidence
that the problem might remain intractable even when the follower
controls only a fixed number of edges. We also show that certain
assumptions on the structure of the problem can be made without loss
of generality, e.g., that the follower controls a tree or that the
leader controls a connected graph.

In view of the negative complexity results mentioned above, one can
expect only very limited positive results. We are able to devise
a~$(|V|-1)$-approximation algorithm for BMST and show that
2-approximating the optimum solution is fixed-parameter tractable in
the number of edges controlled by the follower. For the same
parameter, the decision whether a given follower's response can be
enforced by the leader is fixed-parameter tractable, which implies
that the variant of BMST with uniform costs~$c(e)$ for all~$e\in
E^\ell$ is fixed-parameter tractable as well.  For the bottleneck
case, we show that the problem is tractable in case the leader has a
bottleneck objective and the follower has a sum objective, while it is
hard when the leader has a sum objective and the follower has a
bottleneck objective.  If both have a bottleneck objective, the
problem turns out to be polynomial-time solvable in the pessimistic
setting, while it is hard to solve in the optimistic case.  An
overview of our results for different objective functions can be found
in Table~\ref{tab:bottleneck:results}.  In this paper, however, we
consider a more general variant of BMST than Shi \emph{et
  al.}~\cite{shi19} in terms of the edges controlled by the follower.

The remainder of this paper is organized as follows. In
Section~\ref{sec:restrictions}, we consider different types of
restrictions on the set of allowed instances and investigate their
relations. In Section~\ref{sec:complexity}, we present an
approximation-preserving reduction from Steiner forest to BMST and
derive our main complexity results. Our results concerning
fixed-parameter tractability are presented in Section~\ref{sec:fpt},
while in Section~\ref{sec:algo} we devise an approximation algorithm
for BMST. Up to Section~\ref{sec:algo}, we concentrate on the setting
in which both leader and follower have a sum objective
function. Finally, we review the case of bottleneck objective
functions in Section~\ref{sec:bottleneck}. Section~\ref{sec:conc}
concludes.

\section{Restricted sets of instances}
\label{sec:restrictions}

In this section, we show that, without loss of generality, we may
restrict ourselves to instances of BMST with certain structural
properties. Our aim is to simplify some of the proofs later on, but
also to clarify the connections between different settings
corresponding to reasonable restricted problem variants, which
sometimes lead to different complexity results. All reductions are
polynomial and approximation-preserving, i.e., they can be used to
transform an approximation algorithm for one problem to an
approximation algorithm with the same guarantee for the other problem.

Let $\mathcal{I}$ be the set of all instances $I = (G, E^\ell, E^f, c,
d)$ of BMST as described in the introduction. As already mentioned, we
assume throughout that~$E^\ell$ and~$E^f$ are disjoint sets. If this
is not the case, we can replace any common edge~$e\in E^\ell\cap E^f$
by two parallel edges, one belonging to~$E^\ell$ and one to~$E^f$,
both having the same leader's and follower's costs as~$e$.  We now
define the following subsets of $\mathcal{I}$, all corresponding to
certain restrictions on the edge sets controlled by leader and
follower:
\begin{itemize}
\item $\IElconn$, the set of instances for which the leader's graph $(V,
  E^\ell)$ is connected,
\item $\IElforest$, the set of instances for which the leader's graph
  $(V, E^\ell)$ is cycle-free,
\item $\IEfconn$, the set of instances for which the follower's graph
  $(V, E^f)$ is connected,
\item $\IEfforest$, the set of instances for which the follower's graph
  $(V, E^f)$ is cycle-free,
\item $\IEfmatching$, the set of instances for which the follower's graph
  $(V, E^f)$ is a matching, i.e., for which no vertex is incident to
  more than one edge, and
\item $\IEfE$, the set of instances such that for each leader's edge
  in~$E^\ell$ there exists a parallel follower's edge in~$E^f$ with the
  same leader's cost.
\end{itemize}
The instances in~$\IEfE$ exactly correspond to those considered by~Shi
\emph{et al.}~\cite{shi19}. Since~$G$ is connected, we have
$\IEfE\subseteq\IEfconn$, and $\IEfconn$ is precisely the set of
instances where any cycle-free choice of the leader is feasible, i.e.,
for any cycle-free edge set $X \subseteq E^\ell$, there is at least
one feasible response of the follower. Moreover, we have
$\IEfmatching\subseteq\IEfforest$.

Our first reduction shows that we may assume that the edges controlled
by the follower connect all vertices of~$G$.
\begin{lemma} \label{lem:Efconn} BMST on $\mathcal{I}$ can be reduced
  to BMST on $\IEfconn$. The reduction preserves $\IElconn$,
  $\IElforest$, and $\IEfforest$, meaning that if we start with an
  instance in one of these sets, the reduction again results in an
  instance in this set.
\end{lemma}
\begin{proof}
  Let $I = (G, E^\ell, E^f, c, d) \in \mathcal{I}$. We construct an
  instance $I' \in \IEfconn$ from $I$ by adding arbitrary edges
  controlled by the follower in order to make $(V, E^f)$
  connected. The new edges $e'$ have cost $c(e') \coloneqq d(e')
  \coloneqq M$ for some large enough number $M$, e.g., one can set $M
  \coloneqq \sum_{e \in E} \max\{c(e), d(e)\} + 1$.

  Every solution of $I$ is also a solution of $I'$ of the same cost
  for both leader and follower. The follower's solution is still
  optimum because, by the choice of $M$, taking one of the new edges
  can only make the solution worse for him. Conversely, given a
  solution of $I'$, it is also a solution of $I$ of the same cost if
  it does not contain any new edges. Otherwise, the leader's solution
  of $I'$ is not a feasible choice in $I$, as the follower will only
  take a new edge in $I'$ if he cannot produce all necessary
  connections using only the original edges. In this case, any feasible
  solution of $I$ is cheaper than the one of $I'$, due to the choice
  of~$M$.

  Since $E^\ell$ is not changed, the reduction preserves all
  structural properties of $E^\ell$, in particular~$(V, E^\ell)$ being
  connected or cycle-free.  By adding only a minimum number of edges
  necessary to make $(V,E^f)$ connected, we may also assume that
  acyclicity of~$(V, E^f)$ is preserved.
\end{proof}
Using a similar construction, one can show the same result for the
graph controlled by the leader:
\begin{lemma} \label{lem:Elconn} BMST on $\mathcal{I}$ can be reduced
  to BMST on $\IElconn$. The reduction preserves $\IElforest$,
  $\IEfconn$, $\IEfforest$, and $\IEfmatching$.
\end{lemma}

We next show an important structural result about BMST, from which we
can conclude that we may assume without loss of generality that the
follower controls a forest, but which will also be useful on its
own. As stated in the introduction, we assume a fixed ordering of the
edges in~$E^f$ that the follower will always, i.e., for any choice
of~$X$, use in his greedy algorithm. This is important for the
following proof. Moreover, we do not require $X_1$ or $X_2$
to be feasible leader's solutions in the following, i.e., it might not
be possible for the follower to complete them to a spanning
tree. However, we assume that the follower applies his greedy algorithm
anyway, leading to forests $Y_1$ and $Y_2$.

\begin{lemma} \label{lem:Efforest}
  Given two cycle-free edge sets $X_1 \subseteq X_2 \subseteq
  E^\ell$, let $Y_1, Y_2 \subseteq E^f$ be the corresponding
  follower's responses. Then $Y_1 \supseteq Y_2$.
\end{lemma}
\begin{proof}
  Let~$E^f=\{e_1,\dots,e_m\}$, where~$e_1,\dots,e_m$ is the follower's
  order of preference, and let~$Y_1^{(i)}$ and~$Y_2^{(i)}$ be the
  partial solutions of the follower after considering edge~$e_i$ in
  his greedy algorithm, starting from the leader's choice~$X_1$ or
  $X_2$, respectively. It then suffices to prove the following claim: for
  all~$i=0,\dots,m$, each pair of vertices that is connected
  in~$X_1 \cup Y_1^{(i)}$ is also connected in~$X_2\cup Y_2^{(i)}$. This implies
  that if $e_{i + 1}$ is added to~$Y_2^{(i)}$, it is also added
  to~$Y_1^{(i)}$, so that the full follower's response~$Y_2^{(m)} = Y_2$
  to~$X_2$ is contained in~$Y_1^{(m)}=Y_1$. We show the claim by
  induction over~$i$.

  Since~$Y_1^{(0)}=Y_2^{(0)}=\emptyset$ and~$X_1 \subseteq X_2$, there
  is nothing to show for the case~$i=0$. For~$i = 1,\dots,m$, consider
  two vertices~$v,w\in V$ that are connected by~$X_1 \cup
  Y_1^{(i)}$. If~$v$ and~$w$ are already connected by~$X_1 \cup
  Y_1^{(i-1)}$, they are connected by~$X_2 \cup Y_2^{(i-1)}$ as well,
  by the induction hypothesis, and thus also by the superset~$X_2 \cup
  Y_2^{(i)}$. Otherwise, the connection has been established by
  adding~$e_i=\{v_i,w_i\}$, implying that~$X_1 \cup Y_1^{(i-1)}$
  connects~$v$ to~$v_i$ and~$w$ to~$w_i$ (or vice versa). Again by the
  induction hypothesis, we derive that also~$X_2 \cup Y_2^{(i-1)}$
  connects~$v$ to~$v_i$ and~$w$ to~$w_i$. Hence, either~$v$ and~$w$
  are already connected by~$X_2 \cup Y_2^{(i-1)}$, in which case we
  are done, or~$v_i$ and~$w_i$ are not connected by~$X_2 \cup
  Y_2^{(i-1)}$. In the latter case, edge~$e_i$ will be contained
  in~$Y_2^{(i)}$, so that~$v$ and~$w$ are connected by~$X_2 \cup
  Y_2^{(i)}$ also in this case.  
\end{proof}

\begin{corollary} \label{cor:Efforest}
  BMST on $\mathcal{I}$ can be reduced to BMST on $\IEfforest$. The
  reduction preserves $\IElconn$, $\IElforest$, and $\IEfconn$.
\end{corollary}

\begin{proof}
  Let $I = (G, E^\ell, E^f, c, d) \in \mathcal{I}$ be an instance of
  BMST and let $Y^* \subseteq E^f$ be the result of Kruskal's
  algorithm applied to the graph~$(V, E^f)$, using the fixed order of
  edges defined by the follower's preferences. Note that~$Y^*$ is a
  forest in~$G$, but not necessarily a spanning tree, since we do not
  require~$(V,E^f)$ to be connected. Let~$I'$ be the instance that
  arises from $I$ by removing the edges in~$E^f \setminus Y^*$
  from~$E^f$. Then~$I' \in \IEfforest$. By applying
  Lemma~\ref{lem:Efforest} for $X_1 \coloneqq \emptyset$ and $X_2
  \coloneqq X$, it follows that for any leader's solution~$X$ in~$I$,
  the follower's response~$Y$ lies in~$Y^*$. Hence, $X$ has the same
  objective value in~$I$ as in~$I'$. As the leader's feasible set is
  not changed by the above transformation, we obtain the desired
  reduction result.
  
  Since $E^\ell$ is not changed, the reduction preserves any specific
  structure of~$E^\ell$, in particular $(V, E^\ell)$ being cycle-free
  or connected. Connectedness of~$(V,E^f)$ is obviously
  preserved by the construction.
\end{proof}

It is worth mentioning that even if the follower's edge set $E^f$ is
cycle-free, the follower might have several feasible or even several
optimum responses to some leader's choice~$X$. Indeed, after the
contraction of $X$, the follower's edges might form cycles again. For
an example, consider the instance illustrated in
Fig.~\ref{fig:example}, in which the follower's edges form a
tree. When the leader takes the edge $\{v_2, v_3\}$ into her solution,
the vertices $v_2$ and $v_3$ can be thought of as being merged into a
single vertex from the follower's perspective. This leads to the
follower's edges $\{v_1, v_2\}$ and $\{v_1, v_3\}$ becoming parallel
edges, of which the follower must choose one. In this example, the two
edges even have the same leader's and follower's cost such that the
follower will choose any of the two edges, depending on his
preferences.

If we are not interested in the connectedness of the follower's edge
set~$E^f$, but rather in a simple combinatorial structure of the
latter, we can even further restrict~$E^f$ to form a matching:
\begin{lemma} \label{lem:Efmatching}
  BMST on $\IEfforest$ can be reduced to BMST on $\IEfmatching$. The
  reduction preserves $\IElconn$ and $\IElforest$.
\end{lemma}
\begin{proof}
  Let $I = (G, E^\ell, E^f, c, d) \in \IEfforest$. From $I$, construct
  an instance $I' \in \IEfmatching$ by applying the following
  transformation to every connected component of the graph~$(V, E^f)$
  containing more than one edge: define an arbitrary vertex in the
  connected component as its root. Replace every edge $e \in E^f$ in
  the connected component by a path of length two, with a new vertex
  in the middle. The new edge~$e'$ that is closer to the root, is
  added to $E^\ell$ and assigned $c(e') \coloneqq 0$, while the other
  new edge~$e''$ replaces~$e$ in~$E^f$ and is assigned $d(e'')
  \coloneqq d(e)$ and~$c(e'') \coloneqq c(e)$; moreover, edge~$e''$
  takes the position of~$e$ in the follower's order of
  preference. This construction ensures that~$E^f$ forms a matching
  in~$I'$ because every new vertex has only one incident follower's
  edge, and for every vertex~$v$ that was already present in~$I$, only
  the follower's edge~$e''$ arising from edge~$e$ which is contained
  in the unique path from $v$ to the corresponding root is incident
  to~$v$. See Fig.~\ref{fig:matching} for an illustration.
  \begin{figure}
    \centering
    \small
    \begin{tikzpicture}[scale=1]
      \node[draw,circle,minimum width = 0.6cm] (v1) at (0,0) {$v_1$};
      \node[draw,circle,minimum width = 0.6cm] (v2) at (0,2) {$v_2$};
      \node[draw,circle,minimum width = 0.6cm] (v3) at (2,1) {$v_3$};
      \node[draw,circle,minimum width = 0.6cm] (v4) at (2,3) {$v_4$};
      \node[draw,circle,minimum width = 0.6cm] (v5) at (4,0) {$v_5$};
      \node[draw,circle,minimum width = 0.6cm] (v6) at (4,2) {$v_6$};

      \draw[thick] (v2) to node[above] {$5$} (v3);
      \draw[thick] (v3) to node[right] {$0$} (v4);
      \draw[thick] (v4) to node[above] {$2$} (v6);
      \draw[thick] (v5) to node[right] {$3$} (v6);

      \draw[thick,dashed] (v1) to node[left] {$0 / 2$} (v2);
      \draw[thick,dashed] (v1) to node[below] {~~$0 / 2$} (v3);
      \draw[thick,dashed] (v2) to node[above] {$10 / 0$~~~~} (v4);
      \draw[thick,dashed] (v3) to node[below] {$5 / 1$~~} (v5);
      \draw[thick,dashed] (v3) to node[below] {~~$1 / 0$} (v6);
    \end{tikzpicture}
    \quad\raisebox{1.8cm}{$\leadsto$}\quad
    \begin{tikzpicture}[scale=1]
      \node[draw,circle,minimum width = 0.6cm] (v1) at (0,0) {$v_1$};
      \node[draw,circle,minimum width = 0.5cm] at (0,0) {};
      \node[draw,circle,minimum width = 0.6cm] (v2) at (0,2) {$v_2$};
      \node[draw,circle,minimum width = 0.6cm] (v3) at (2,1) {$v_3$};
      \node[draw,circle,minimum width = 0.6cm] (v4) at (2,3) {$v_4$};
      \node[draw,circle,minimum width = 0.6cm] (v5) at (4,0) {$v_5$};
      \node[draw,circle,minimum width = 0.6cm] (v6) at (4,2) {$v_6$};

      \draw[thick] (v2) to node[above] {$5$} (v3);
      \draw[thick] (v3) to node[right] {$0$} (v4);
      \draw[thick] (v4) to node[above] {$2$} (v6);
      \draw[thick] (v5) to node[right] {$3$} (v6);

      \node[draw,circle,minimum width = 0.15cm] (w12) at (0,1) {};
      \node[draw,circle,minimum width = 0.15cm] (w24) at (1,2.5) {};
      \node[draw,circle,minimum width = 0.15cm] (w13) at (1,0.5) {};
      \node[draw,circle,minimum width = 0.15cm] (w36) at (3,1.5) {};
      \node[draw,circle,minimum width = 0.15cm] (w35) at (3,0.5) {};

      \draw[thick] (v1) to node[left] {$0$} (w12);
      \draw[thick,dashed] (v2) to node[left] {$0 / 2$} (w12);
      \draw[thick] (v1) to node[below] {$0$} (w13);
      \draw[thick,dashed] (v3) to node[below] {~~$0 / 2$} (w13);
      \draw[thick] (v2) to node[above] {$0$} (w24);
      \draw[thick,dashed] (v4) to node[above] {$10 / 0$~~~~} (w24);
      \draw[thick] (v3) to node[below] {$0$} (w35);
      \draw[thick,dashed] (v5) to node[below] {$5 / 1$~~} (w35);
      \draw[thick] (v3) to node[below] {$0$} (w36);
      \draw[thick,dashed] (v6) to node[below] {~~$1 / 0$} (w36);

    \end{tikzpicture}
    \caption{Illustration of the construction in the proof of
      Lemma~\ref{lem:Efmatching} applied to the example instance of
      Fig.~\ref{fig:example}, with vertex~$v_1$ chosen as root.}
    \label{fig:matching}
  \end{figure}
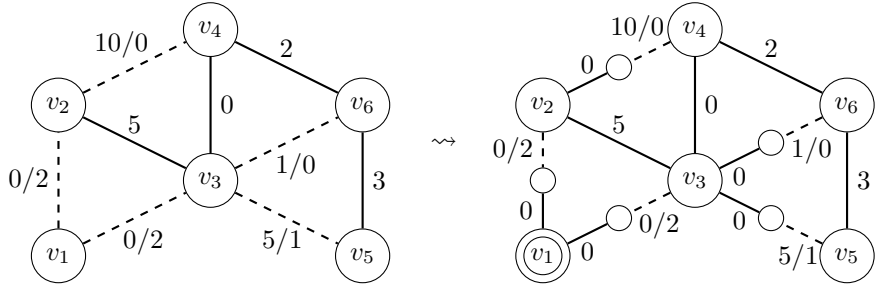

  A solution for $I$ can be transformed to a solution for $I'$ of the
  same cost by adding all newly introduced leader's edges to her
  solution, which does not change the cost. Indeed, the follower
  solves exactly the same problem after the leader's solution is
  contracted. For the opposite transformation, consider an optimum
  leader's solution~$X'$ for $I'$. Observe that we may assume all
  newly introduced edges to be in~$X'$ because otherwise, adding them
  would lead to the follower removing some of his edges from his
  response by Lemma~\ref{lem:Efforest}, which cannot increase the
  leader's objective value. Now, remove all new edges from~$X'$ in
  order to get a solution~$X$ for~$I$. Again, the follower has exactly
  the same choices responding to $X$ and $X'$, respectively. Thus, the
  objective value of~$X$ in~$I$ is at most the objective value of~$X'$
  in~$I'$.
  
  Since the edges added to $E^\ell$ connect every new vertex by
  exactly one edge, the reduction preserves $\IElconn$ and
  $\IElforest$.
\end{proof}
Combining the reductions from Lemma~\ref{lem:Efconn},
Lemma~\ref{lem:Elconn}, and Corollary~\ref{cor:Efforest}, we obtain the
following result:
\begin{corollary} \label{cor:Eftree}
  BMST on $\mathcal{I}$ can be reduced to BMST on
  $\IElconn\cap\IEftree$. The reduction preserves~$\IElforest$.
\end{corollary}

Dropping the connectedness of $E^f$, we can apply
Lemma~\ref{lem:Efmatching} to obtain:
\begin{corollary} \label{cor:Efmatching}
  BMST on $\mathcal{I}$ can be reduced to BMST on
  $\IElconn\cap\IEfmatching$. The reduction preserves~$\IElforest$.
\end{corollary}

As mentioned above, the authors of~\cite{shi19} only consider
instances from~$\IEfE$, i.e., the follower controlling many
edges. This could be seen as an opposite assumption to instances being
chosen from~$\IEfforest$ or even~$\IEfmatching$. To show that our main
complexity results still hold in the setting of~\cite{shi19}, we use
the following result:
\begin{lemma} \label{lem:EfE}
  BMST on $\IEfconn$ can be reduced to BMST on $\IEfE$.
  The reduction preserves $\IElconn$ and $\IElforest$.
\end{lemma}
\begin{proof}
  Let $I = (G, E^\ell, E^f, c, d) \in \IEfconn$. Construct an instance
  $I' \in \IEfE$ from~$I$ by creating a copy $e'$ of each edge $e \in
  E^\ell$ that does not have a parallel follower's edge of the same
  leader's cost, adding $e'$ to $E^f$ and setting $c(e') \coloneqq c(e)$ and
  $d(e') \coloneqq M$, for some large $M$, e.g., $M \coloneqq \sum_{e \in E} d(e) +
  1$. The construction is illustrated in Fig.~\ref{fig:all}.
  \begin{figure}
    \centering
    \small
    \begin{tikzpicture}[scale=1]
      \node[draw,circle,minimum width = 0.6cm] (v1) at (0,0) {$v_1$};
      \node[draw,circle,minimum width = 0.6cm] (v2) at (0,2) {$v_2$};
      \node[draw,circle,minimum width = 0.6cm] (v3) at (2,1) {$v_3$};
      \node[draw,circle,minimum width = 0.6cm] (v4) at (2,3) {$v_4$};
      \node[draw,circle,minimum width = 0.6cm] (v5) at (4,0) {$v_5$};
      \node[draw,circle,minimum width = 0.6cm] (v6) at (4,2) {$v_6$};

      \draw[thick] (v2) to node[above] {$5$} (v3);
      \draw[thick] (v3) to node[right] {$0$} (v4);
      \draw[thick] (v4) to node[above] {$2$} (v6);
      \draw[thick] (v5) to node[right] {$3$} (v6);

      \draw[thick,dashed] (v1) to node[left] {$0 / 2$} (v2);
      \draw[thick,dashed] (v1) to node[below] {~~$0 / 2$} (v3);
      \draw[thick,dashed] (v2) to node[above] {$10 / 0$~~~~} (v4);
      \draw[thick,dashed] (v3) to node[below] {$5 / 1$~~} (v5);
      \draw[thick,dashed] (v3) to node[below] {~~$1 / 0$} (v6);
    \end{tikzpicture}
    \quad\raisebox{1.8cm}{$\leadsto$}\quad
    \begin{tikzpicture}[scale=1]
      \node[draw,circle,minimum width = 0.6cm] (v1) at (0,0) {$v_1$};
      \node[draw,circle,minimum width = 0.6cm] (v2) at (0,2) {$v_2$};
      \node[draw,circle,minimum width = 0.6cm] (v3) at (2,1) {$v_3$};
      \node[draw,circle,minimum width = 0.6cm] (v4) at (2,3) {$v_4$};
      \node[draw,circle,minimum width = 0.6cm] (v5) at (4,0) {$v_5$};
      \node[draw,circle,minimum width = 0.6cm] (v6) at (4,2) {$v_6$};

      \draw[thick,transform canvas={xshift=1pt,yshift=2pt}] (v2) to node[above] {$5$} (v3);
      \draw[thick,dashed,transform canvas={xshift=-1pt,yshift=-2pt}] (v2) to node[below] {$5 / M$~~~~} (v3);
      \draw[thick,transform canvas={xshift=2.5pt,yshift=0pt}] (v3) to node[right,yshift=-5pt] {$0$} (v4);
      \draw[thick,dashed,transform canvas={xshift=-2.5pt,yshift=0pt}] (v3) to node[left,yshift=5pt] {$0 / M$} (v4);
      \draw[thick,transform canvas={xshift=1pt,yshift=2pt}] (v4) to node[above] {$2$} (v6);
      \draw[thick,dashed,transform canvas={xshift=-1pt,yshift=-2pt}] (v4) to node[below] {$2 / M$~~~~} (v6);
      \draw[thick,transform canvas={xshift=2.5pt,yshift=0pt}] (v5) to node[right] {$3$} (v6);
      \draw[thick,dashed,transform canvas={xshift=-2.5pt,yshift=0pt}] (v5) to node[left,yshift=-5pt] {$3 / M$} (v6);

      \draw[thick,dashed] (v1) to node[left] {$0 / 2$} (v2);
      \draw[thick,dashed] (v1) to node[below] {~~$0 / 2$} (v3);
      \draw[thick,dashed] (v2) to node[above] {$10 / 0$~~~~} (v4);
      \draw[thick,dashed] (v3) to node[below] {$5 / 1$~~} (v5);
      \draw[thick,dashed] (v3) to node[below] {~~$1 / 0$} (v6);
    \end{tikzpicture}
    \caption{Illustration of the construction in the proof of
      Lemma~\ref{lem:EfE} applied to the example instance of
      Fig.~\ref{fig:example}.}
    \label{fig:all}
  \end{figure}
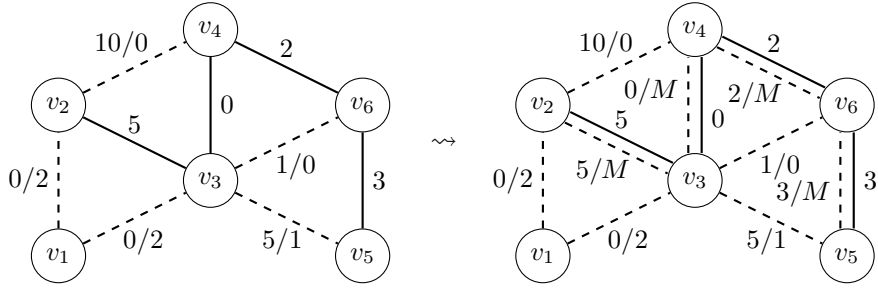

  All cycle-free sets~$X \subseteq E^\ell$ are feasible leader's
  solutions for both $I$ and $I'$ because we assume $(V, E^f)$ to be
  connected. By construction, any feasible leader's solution~$X$ leads
  to the same follower's response in~$I$ and~$I'$, since the
  additional edges are the most expensive ones for the follower and
  will thus never be chosen because he can establish any desired
  connection using only the original edges.

  As the reduction does not change the set $E^\ell$, it clearly
  preserves all its structural properties, in particular $(V, E^\ell)$
  being cycle-free or connected.
\end{proof}
From Lemma~\ref{lem:Efconn} and Lemma~\ref{lem:EfE} we derive
\begin{corollary}\label{cor:EfE}
  BMST on $\mathcal{I}$ can be reduced to BMST on $\IEfE$. The
  reduction preserves $\IElconn$ and $\IElforest$.
\end{corollary}
In the following sections, we will also consider the case of uniform
leader's costs on the leader's edges~$E^\ell$. For this, we show
\begin{lemma} \label{lem:uniform}
  BMST on $\IEfconn$ with polynomially-bounded integer costs $c$ on
  $E^\ell$ can be reduced to BMST with $c(e) = 1$ for all $e \in
  E^\ell$. The reduction preserves $\IElconn$, $\IElforest$,
  $\IEfconn$, and $\IEfforest$.
\end{lemma}
\begin{proof}
  Let $I = (G, E^\ell, E^f, c, d) \in \IEfconn$ with
  polynomially-bounded integer costs~$c$ on $E^\ell$.  Construct an
  instance~$I'$ of BMST with uniform costs $c$ on $E^\ell$ as follows:
  contract all edges $e \in E^\ell_0 \coloneqq \{e \in E^\ell \mid
  c(e) = 0\}$. Each edge~$e=\{v,w\}\in E^\ell \setminus E^\ell_0$ is
  replaced by a path~$P_e$ of length~$c(e)$, consisting of leader's
  edges again.  Each interior vertex~$u$ of $P_e$ is connected to~$v$
  by a new edge~$e'$ added to~$E^f$ with $c(e')\coloneqq0$ and
  $d(e')\coloneqq M$ for some large enough constant~$M \coloneqq
  \sum_{e \in E} d(e) +1$. Note that for edges~$e \in E^\ell$
  with~$c(e) = 1$ nothing changes.

  We claim that the instances~$I$ and~$I'$ have the same optimum
  value. Given an optimum solution~$X$ to~$I$, we first may assume
  that $X$ contains a maximal forest in~$E^\ell_0$ because otherwise,
  we could add an edge from $E^\ell_0$ to $X$, replacing some edge~$e$
  with $c(e) \geq 0$ in the resulting spanning tree. A feasible
  solution~$X'$ to~$I'$ having the same objective value as $X$ can be
  defined by setting~$X'\coloneqq\cup_{e\in X}P_e$.  This is true
  since the follower has to connect all interior vertices of
  paths~$P_e$ with~$e\not\in X$ using the newly introduced follower's
  edges in order to ensure that the resulting graph is a tree. These
  edges have cost $0$ for the leader.  After adding these edges, the
  follower has exactly the same choices as in the instance~$I$.

  Conversely, given an optimum solution~$X'$ to~$I'$, we may assume
  that, for each edge~$e\in E^\ell \setminus E^\ell_0$, either all
  edges in~$P_e$ belong to~$X'$ or none: assume this is not true and
  consider some solution~$X'$ to~$I'$ that contradicts this property.
  Let~$Y'$ be the follower's response to~$X'$. We construct a
  solution~$X'' \subset X'$ to~$I'$ with follower's response $Y''$
  with $c(X'') < c(X')$ and $c(Y') = c(Y'')$ as follows: let
  $$X'' \coloneqq \bigcup \{P_e \mid e \in
  E^\ell \setminus E^\ell_0, P_e \subseteq X'\}$$ consist of all the
  paths that are entirely contained in $X'$, i.e., we simply leave out
  all edges of paths that were only taken partially in $X'$.  As we
  assume $I \in \IEfconn$, which implies also $I' \in \IEfconn$, the
  leader's solution $X''$ is clearly feasible because the follower can
  complete any solution to a spanning tree. Moreover, observe that,
  since the edges connecting the inner vertices of the paths~$P_e$
  have very high cost for the follower, they are only taken if
  absolutely necessary. Therefore, the response~$Y''$ to $X''$ is the
  same as the response $Y'$ to $X'$ with some additional edges that
  connect the inner vertices of the paths $P_e$ that are connected by
  $X'$, but not by $X''$. This shows $c(Y') = c(Y'')$, since these
  additional edges have cost $0$ for the leader; hence, $X''$ is the
  desired solution. Thus, we can assume that, for each~$e\in E^\ell
  \setminus E^\ell_0$, either all edges in~$P_e$ belong to~$X'$ or
  none.  Setting~$X\coloneqq\{e\in E^\ell \setminus E^\ell_0 \mid
  P_e\subseteq X'\} \cup F$, where $F$ is a maximal forest in
  $E^\ell_0$, then yields a feasible solution to~$I$ with the same
  objective value as~$X'$ in~$I'$ because the follower's responses to
  $X$ and $X'$ have the same cost, by the same arguments as in the
  first part of the proof.

  Acyclicity and connectedness of both $E^\ell$ and $E^f$ are
  preserved because the construction ensures that the newly introduced
  vertices are all connected to the old vertices in an acyclic manner
  in both $E^\ell$ and $E^f$.
\end{proof}
The reduction described in the proof of Lemma~\ref{lem:Efconn} only
introduces follower's edges. We can thus combine it with
Lemma~\ref{lem:uniform} to obtain
\begin{corollary}\label{cor:uniform}
  BMST on ${\mathcal I}$ with polynomially-bounded integer costs $c$
  on $E^\ell$ can be reduced to BMST with $c(e) = 1$ for all $e \in
  E^\ell$. The reduction preserves $\IElconn$, $\IElforest$,
  $\IEfconn$, and $\IEfforest$.
\end{corollary}
  
\section{Main complexity results}
\label{sec:complexity}

In this section, we establish a first hardness result for BMST using
a reduction from the well-known \emph{Steiner forest problem}:

\medskip

(SF) Given a connected graph $G = (V, E)$ with edge lengths
$\ell\colon E \rightarrow \mathbb{R}_{\geq 0}$ and~$k$~disjoint sets
$S_1, \dots, S_k \subseteq V$, find a forest~$F\subseteq E$ of minimum
total length~$\ell(F)$, such that for each terminal set~$S_i$, all vertices in~$S_i$ are connected
in the graph~$(V, F)$.

\medskip

The best approximation ratio that is known for SF is $2$~\cite{jain2001factor}
and the problem is NP-hard to approximate within a factor of
$\tfrac{96}{95}$~\cite{chlebik08}.  We will reduce SF to BMST in
order to obtain the following result:
\begin{theorem}\label{thm:hard}
  BMST cannot be approximated to within a factor of $\frac{96}{95}$ in
  polynomial time, unless~P$\,=\,$NP, even if~$E^f$ is a tree.
\end{theorem}
\begin{proof}
  Let $I$ be an instance of SF, consisting of a graph $G = (V, E)$
  with edge lengths $\ell\colon E \rightarrow \mathbb{R}_{\geq 0}$ and
  disjoint terminal sets $S_1, \dots, S_k \subseteq V$.  We construct
  an instance~$I'$ of BMST as follows.  The graph in $I'$ is $G'
  \coloneqq (V, E^\ell \cup E^f )$, where~$E^\ell\coloneqq E$
  and~$E^f$ is defined as follows: first introduce edges forming any
  forest with connected components having vertex sets $S_1, \dots,
  S_k$ and call this edge set $E^f_0$. Then, add any further edges
  turning~$E^f_0$ into a spanning tree on~$V$. All new edges together
  form the set $E^f$. The cost function for the leader is
    \[
        c(e) \coloneqq
        \begin{cases}
            \ell(e),	&	\text{if } e \in E^\ell ,\\
            M,	&	\text{if } e \in E^f_0 ,\\
            0, &	\text{if } e \in E^f \setminus E^f_0,
        \end{cases}
    \] 
    where $M$ is some large constant such as $\sum_{e \in E} \ell(e) +1$.
    The cost function for the follower is given by
    \[
        d(e) \coloneqq
        \begin{cases}
            0,	&	\text{if } e \in E^f_0 ,\\
            1, &	\text{if } e \in E^f \setminus E^f_0.
        \end{cases}
    \] 
    This finishes the construction of~$I'$.  We now show that any
    optimum solution~$X$ to $I$ corresponds to a feasible leader's
    solution~$X'$ to $I'$ of the same cost, and vice versa. So let $X
    \subseteq E$ be any solution to $I$.  Then $X' \coloneqq X$ is a
    feasible leader's solution since~$X$ forms a forest and $E^f$
    connects all vertices, so that the follower can complete any
    leader's solution to a tree. Since $X$ connects each terminal set,
    the follower's response to $X'$ does not contain any edges from
    $E^f_0$ as they would form a cycle together with~$X'$. Hence, the
    follower's response only consists of edges having cost $0$ for
    the leader. Therefore, the overall cost for the leader is
    simply~$c(X') = \ell(X)$.

    It remains to show that any optimum solution $X'$ to $I'$
    corresponds to a feasible solution $X$ to SF of the same
    cost. Clearly, there exists a leader's solution to~$I'$ of cost at
    most~$M-1$, e.g., one could choose any spanning tree
    in~$G=(V,E^\ell)$. By optimality of~$X'$, this implies that the
    follower's response to~$X'$ does not contain any of the edges in
    $E^f_0$. However, since the follower's cost for the edges in
    $E^f_0$ is cheaper than the cost of the edges in $E^f \setminus
    E^f_0$, this implies that the leader's solution~$X'$ connects each
    terminal set. As $X'$ is also cycle-free, it is a solution to~$I$
    having cost~$\ell(X')$. 
\end{proof}
\begin{remark}\label{rem:treestructure}
  The definition of $E^f$ in the proof of Theorem~\ref{thm:hard}
  leaves a lot of freedom concerning the structure of the follower's
  tree. For example, it can always be chosen to form a path. Moreover,
  the reduction can be performed analogously from the Steiner tree
  problem instead of the Steiner forest problem, i.e., where only one
  terminal set $S_1$ is given. Then the structure of the follower's
  tree is even less restricted, for example, the set~$E^f$ can be
  chosen to form a star; see Fig.~\ref{fig:steiner} for an
  illustration. Thus, the hardness of Theorem~\ref{thm:hard} still
  holds for restrictions of the follower's tree's structure such as
  $E^f$ being a path or a star.
\end{remark}
\begin{figure}
  \centering
  \small
  \begin{tikzpicture}[scale=1]
    \node[draw,circle,minimum width = 0.6cm] (v1) at (0,0) {$v_1$};
    \node[draw,circle,minimum width = 0.5cm] at (0,0) {};
    \node[draw,circle,minimum width = 0.6cm] (v2) at (0,2) {$v_2$};
    \node[draw,circle,minimum width = 0.6cm] (v3) at (2,1) {$v_3$};
    \node[draw,circle,minimum width = 0.6cm] (v4) at (2,3) {$v_4$};
    \node[draw,circle,minimum width = 0.5cm] at (2,3) {};
    \node[draw,circle,minimum width = 0.6cm] (v5) at (4,0) {$v_5$};
    \node[draw,circle,minimum width = 0.6cm] (v6) at (4,2) {$v_6$};
    \node[draw,circle,minimum width = 0.5cm] at (4,2) {};

    \draw[thick] (v2) to node[below] {$5$} (v3);
    \draw[thick,red] (v3) to node[left] {$0$} (v4);
    \draw[thick] (v4) to node[above] {$2$} (v6);
    \draw[thick] (v5) to node[right] {$3$} (v6);
    \draw[thick] (v1) to node[left] {$0$} (v2);
    \draw[thick,red] (v1) to node[below] {$0$} (v3);
    \draw[thick] (v2) to node[above] {$10$} (v4);
    \draw[thick] (v3) to node[below] {$5$} (v5);
    \draw[thick,red] (v3) to node[below] {$1$} (v6);
  \end{tikzpicture}
  \quad\raisebox{1.8cm}{$\leadsto$}\quad
  \begin{tikzpicture}[scale=1]
    \node[draw,circle,minimum width = 0.6cm] (v1) at (0,0) {$v_1$};
    \node[draw,circle,minimum width = 0.5cm] at (0,0) {};
    \node[draw,circle,minimum width = 0.6cm] (v2) at (0,2) {$v_2$};
    \node[draw,circle,minimum width = 0.6cm] (v3) at (2,1) {$v_3$};
    \node[draw,circle,minimum width = 0.6cm] (v4) at (2,3) {$v_4$};
    \node[draw,circle,minimum width = 0.5cm] at (2,3) {};
    \node[draw,circle,minimum width = 0.6cm] (v5) at (4,0) {$v_5$};
    \node[draw,circle,minimum width = 0.6cm] (v6) at (4,2) {$v_6$};
    \node[draw,circle,minimum width = 0.5cm] at (4,2) {};

    \draw[thick] (v2) to node[pos=0.7,below] {$5$} (v3);
    \draw[thick,transform canvas={xshift=-2pt,yshift=0pt},red] (v3) to node[pos=0.3,left] {$0$} (v4);
    \draw[thick,transform canvas={xshift=1pt,yshift=2pt}] (v4) to node[above] {$2$} (v6);
    \draw[thick] (v5) to node[right] {$3$} (v6);
    \draw[thick] (v1) to node[left] {$0$} (v2);
    \draw[thick,red] (v1) to node[below] {$0$} (v3);
    \draw[thick,transform canvas={xshift=-1pt,yshift=2pt}] (v2) to node[above] {$10$} (v4);
    \draw[thick] (v3) to node[below] {$5$} (v5);
    \draw[thick,red] (v3) to node[pos=0.3,below] {$1$} (v6);

    \draw[thick,dashed] (v1) to node[pos=0.4,left] {$M/0$} (v4);
    \draw[thick,dashed,transform canvas={xshift=0.75pt,yshift=-1.5pt}] (v4) to node[pos=0.7,below] {$M/0$~~~} (v6);

    \draw[very thick,dotted,transform canvas={xshift=2pt,yshift=0pt}] (v3) to node[pos=0.3,right] {$0/1$} (v4);
    \draw[very thick,dotted,transform canvas={xshift=0.75pt,yshift=-1.5pt},red] (v2) to node[pos=0.35,below] {$0/1$} (v4);
    \draw[very thick,dotted,transform canvas={xshift=2.5pt,yshift=0pt},red] (v4) to node[pos=0.6,right] {$0/1$} (v5);

  \end{tikzpicture}
  \caption{Illustration of the proof of Theorem~\ref{thm:hard}
    for~$k=1$ with~$E_f$ being a star. The marked
    vertices~$v_1,v_4,v_6$ are the given terminals. Edges in~$E^f_0$ are
    represented as dashed edges, the remaining edges in~$E^f$ are
    drawn as dotted egdes. Any vertex of~$S_1$ can be chosen as the
    center of the star, here it is~$v_4$. Red edges mark the optimum
    Steiner tree in the input graph and the optimum leader's solution
    and corresponding follower's response in the BMST instance.}
  \label{fig:steiner}
\end{figure}
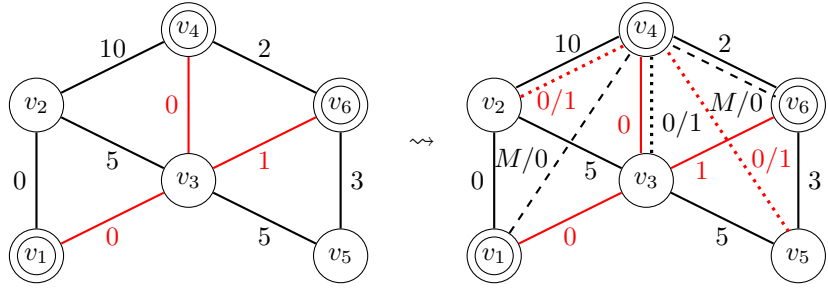

Theorem~\ref{thm:hard} and Corollary~\ref{cor:EfE}
together prove a conjecture stated by Shi \emph{et al.}~\cite{shi19}:
\begin{corollary}\label{cor:hard:EfE}
  BMST on $\IEfE$ cannot be approximated to within a factor
  of~$\frac{96}{95}$ in polynomial time, unless~P$\,=\,$NP.
\end{corollary}

\begin{remark} \label{rem:mmst}
  If we allow negative costs in BMST, the proof of
  Theorem~\ref{thm:hard} works in the same way if we define $d(e) \coloneqq
  -c(e)$ for all $e \in E^\ell \cup E^f$ instead. This shows that the
  special case of BMST in which the follower is adversarial to the
  leader, having the opposite objective function, is hard as
  well. This is in contrast to~\cite{shi19} where this special case
  (called MMST there) is shown to be polynomial-time solvable, for
  both sum and bottleneck objective. However, this is not a
  contradiction because the authors of~\cite{shi19} only work with
  instances from $\IEfE$. In fact, Corollary~\ref{cor:hard:EfE} does
  not carry over to the special case of MMST since the property of
  opposite objective functions is lost in the construction in
  Lemma~\ref{lem:EfE}.
\end{remark}

Together with Corollary~\ref{cor:Efmatching}, we can conclude that
BMST remains hard even if the follower controls a matching, and hence
a very simple combinatorial structure.
\begin{corollary}
  \label{cor:hard:Efmatching}
  BMST cannot be approximated to within a factor of $\frac{96}{95}$ in
  polynomial time, unless~P$\,=\,$NP, even if~$E^f$ is a matching.
\end{corollary}

From Corollary~\ref{cor:uniform} it follows that the hardness of BMST is
preserved even in the case of uniform leader's costs on her own
edges. We emphasize that Theorem~\ref{thm:hard} still holds for
polynomially-bounded and integer leader's cost since Steiner forest is
strongly NP-hard~\cite{bern1989steiner}.
\begin{corollary}
  \label{cor:hard:uniform}
  BMST cannot be approximated to within a factor of $\frac{96}{95}$ in
  polynomial time, unless~P$\,=\,$NP, even if $E^f$ is a tree
  and~$c(e)=1$ holds for all~$e\in E^\ell$.
\end{corollary}

To conclude this section, we consider a related question which could
be asked in any bilevel optimization problem: can the leader enforce a
given follower's response? More formally, we consider the following
decision problem:

\medskip

(BMST-R) Given an instance of BMST and a set~$\bar Y\subseteq E^f$,
does there exist some leader's choice~$X\subseteq E^\ell$ such
that~$\bar Y$ is the follower's response to~$X$?

\medskip

For this problem to be well-defined, as for BMST itself, it
is essential to assume that the follower has a consistent strategy to
select a follower's response in case his optimum solution is not
unique. As discussed in the introduction, we ensure such a consistent
strategy by assuming that the follower chooses edges greedily
according to some deterministic order.

Apart from being an interesting structural question in its own right,
we will see in Section~\ref{sec:fpt} that BMST-R -- or more precisely,
the optimization version in which the cheapest solution~$X$
enforcing~$\bar Y$ is desired -- is related to the fixed-parameter
tractability of BMST in terms of~$|E^f|$. However, we will prove that
BMST-R, even in the decision version, is NP-complete. For this, we use
the so-called \emph{vertex-disjoint Steiner trees} problem:

\medskip

(VDST) Given a connected graph~$G=(V,E)$ and~$k$~disjoint
sets~$S_1,\dots,S_k\subseteq V$, do there exist vertex-disjoint
trees~$T_1,\dots,T_k\subseteq E$ in~$G$ such that $T_i$ spans $S_i$
for all~$i=1,\dots,k$?

\medskip

This problem is similar to the Steiner forest problem
defined previously, but not the same. The important difference is that
in the Steiner forest problem, no disjointness of the trees in the
solution is required, i.e., it is feasible to have several sets $S_i$
lying in the same connected component of the solution. Moreover, we
are considering the decision version of the vertex-disjoint Steiner
trees problem here, without any edge costs. Such a decision version of
Steiner forest would not be interesting because it is always feasible
to select a spanning tree.

The problem VDST is known to be NP-complete even for $k=2$ in
so-called two-layer routing graphs~\cite{korte1990}. We use this fact
to prove the following result:

\begin{theorem}\label{thm:hard:response}
  BMST-R is NP-complete, even if~$|\bar Y|=1$ and $E^f$ forms a path
  on a subset of the vertex set.
\end{theorem}
\begin{proof}
  BMST-R clearly belongs to NP. To show completeness, we reduce VDST
  for~$k=2$ to BMST-R. Given an instance of VDST consisting of a
  connected graph~$G=(V,E)$ and disjoint sets $S=\{s_1,\dots,s_{r}\}$
  and~$S'=\{s'_1,\dots,s'_{r'}\}$, we define an instance of BMST-R on
  $V$ by setting~$E^\ell\coloneqq E$ and
  $$E^f\coloneqq \big\{\{s_i,s_{i+1}\}\mid i=1,\dots,r-1\big\}\cup
  \big\{\{s'_i,s'_{i+1}\}\mid
  i=1,\dots,r'-1\big\}\cup\big\{\{s_1,s'_1\}\big\}\;,$$
  where~$d(\{s_1,s'_1\})\coloneqq 1$ and~$d(e)\coloneqq 0$ for all~$e\in
  E^f\setminus\big\{\{s_1,s'_1\}\big\}$. Let~$\bar
  Y\coloneqq \big\{\{s_1,s'_1\}\big\}$. The leader's cost function $c$ is
  irrelevant for the problem BMST-R.
  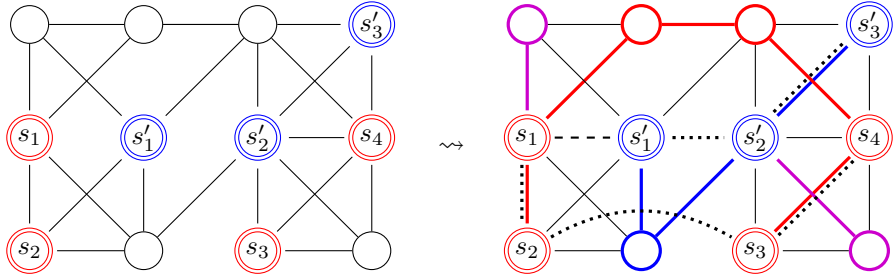
\begin{figure}
  \centering
  \small
  \begin{tikzpicture}[scale=0.75]
    \node[circle,minimum width = 0.5cm] (v1) at (0,0) {$s_2$};
    \node[red,draw,circle,minimum width = 0.6cm] at (0,0) {};
    \node[red,draw,circle,minimum width = 0.5cm] at (0,0) {};
    \node[circle,minimum width = 0.5cm] (v2) at (0,2) {$s_1$};
    \node[red,draw,circle,minimum width = 0.6cm] at (0,2) {};
    \node[red,draw,circle,minimum width = 0.5cm] at (0,2) {};
    \node[draw,circle,minimum width = 0.5cm] (v3) at (0,4) {};
    \node[draw,circle,minimum width = 0.5cm] (v4) at (2,0) {};
    \node[circle,minimum width = 0.5cm] (v5) at (2,2) {$s_1'$};
    \node[blue,draw,circle,minimum width = 0.6cm] at (2,2) {};
    \node[blue,draw,circle,minimum width = 0.5cm] at (2,2) {};
    \node[draw,circle,minimum width = 0.5cm] (v5b) at (2,4) {};
    \node[circle,minimum width = 0.5cm] (v6) at (4,0) {$s_3$};
    \node[red,draw,circle,minimum width = 0.6cm] at (4,0) {};
    \node[red,draw,circle,minimum width = 0.5cm] at (4,0) {};
    \node[circle,minimum width = 0.5cm] (v7) at (4,2) {$s_2'$};
    \node[blue,draw,circle,minimum width = 0.6cm] at (4,2) {};
    \node[blue,draw,circle,minimum width = 0.5cm] at (4,2) {};
    \node[draw,circle,minimum width = 0.5cm] (v8) at (4,4) {};
    \node[draw,circle,minimum width = 0.5cm] (v9) at (6,0) {};
    \node[circle,minimum width = 0.5cm] (v10) at (6,2) {$s_4$};
    \node[red,draw,circle,minimum width = 0.6cm] at (6,2) {};
    \node[red,draw,circle,minimum width = 0.5cm] at (6,2) {};
    \node[circle,minimum width = 0.5cm] (v11) at (6,4) {$s_3'$};
    \node[blue,draw,circle,minimum width = 0.6cm] at (6,4) {};
    \node[blue,draw,circle,minimum width = 0.5cm] at (6,4) {};

    \draw (v1) to (v2);
    \draw (v1) to (v4);
    \draw (v1) to (v5);
    \draw (v2) to (v3);
    \draw (v2) to (v4);
    \draw (v2) to (v5b);
    \draw (v3) to (v5);
    \draw (v3) to (v5b);
    \draw (v4) to (v5);
    \draw (v4) to (v7);
    \draw (v5) to (v8);
    \draw (v5b) to (v8);
    \draw (v5b) to (v8);
    \draw (v6) to (v7);
    \draw (v6) to (v9);
    \draw (v6) to (v10);
    \draw (v7) to (v8);
    \draw (v7) to (v9);
    \draw (v7) to (v10);
    \draw (v7) to (v11);
    \draw (v8) to (v10);
    \draw (v8) to (v11);
    \draw (v9) to (v10);
    \draw (v10) to (v11);
  \end{tikzpicture}
  \quad\raisebox{1.66cm}{$\leadsto$}~\quad
  \begin{tikzpicture}[scale=0.75]
    \definecolor{dpurple}{rgb}{0.8, 0.0, 0.8};
    \node[circle,minimum width = 0.5cm] (v1) at (0,0) {$s_2$};
    \node[red,draw,circle,minimum width = 0.6cm] at (0,0) {};
    \node[red,draw,circle,minimum width = 0.5cm] at (0,0) {};
    \node[circle,minimum width = 0.5cm] (v2) at (0,2) {$s_1$};
    \node[red,draw,circle,minimum width = 0.6cm] at (0,2) {};
    \node[red,draw,circle,minimum width = 0.5cm] at (0,2) {};
    \node[dpurple,draw,circle,minimum width = 0.5cm,very thick] (v3) at (0,4) {};
    \node[blue,draw,circle,minimum width = 0.5cm,very thick] (v4) at (2,0) {};
    \node[circle,minimum width = 0.5cm] (v5) at (2,2) {$s_1'$};
    \node[blue,draw,circle,minimum width = 0.6cm] at (2,2) {};
    \node[blue,draw,circle,minimum width = 0.5cm] at (2,2) {};
    \node[red,draw,circle,minimum width = 0.5cm,very thick] (v5b) at (2,4) {};
    \node[circle,minimum width = 0.5cm] (v6) at (4,0) {$s_3$};
    \node[red,draw,circle,minimum width = 0.6cm] at (4,0) {};
    \node[red,draw,circle,minimum width = 0.5cm] at (4,0) {};
    \node[circle,minimum width = 0.5cm] (v7) at (4,2) {$s_2'$};
    \node[blue,draw,circle,minimum width = 0.6cm] at (4,2) {};
    \node[blue,draw,circle,minimum width = 0.5cm] at (4,2) {};
    \node[red,draw,circle,minimum width = 0.5cm,very thick] (v8) at (4,4) {};
    \node[dpurple,draw,circle,minimum width = 0.5cm,very thick] (v9) at (6,0) {};
    \node[circle,minimum width = 0.5cm] (v10) at (6,2) {$s_4$};
    \node[red,draw,circle,minimum width = 0.6cm] at (6,2) {};
    \node[red,draw,circle,minimum width = 0.5cm] at (6,2) {};
    \node[circle,minimum width = 0.5cm] (v11) at (6,4) {$s_3'$};
    \node[blue,draw,circle,minimum width = 0.6cm] at (6,4) {};
    \node[blue,draw,circle,minimum width = 0.5cm] at (6,4) {};

    \draw[red,very thick] (v1) to (v2);
    \draw (v1) to (v4);
    \draw (v1) to (v5);
    \draw[dpurple,very thick] (v2) to (v3);
    \draw (v2) to (v4);
    \draw[red,very thick] (v2) to (v5b);
    \draw (v3) to (v5);
    \draw (v3) to (v5b);
    \draw[blue,very thick] (v4) to (v5);
    \draw[blue,very thick] (v4) to (v7);
    \draw (v5) to (v8);
    \draw[red,very thick] (v5b) to (v8);
    \draw (v6) to (v7);
    \draw (v6) to (v9);
    \draw[red,very thick] (v6) to (v10);
    \draw (v7) to (v8);
    \draw[dpurple,very thick] (v7) to (v9);
    \draw (v7) to (v10);
    \draw[blue,very thick] (v7) to (v11);
    \draw[red,very thick] (v8) to (v10);
    \draw (v8) to (v11);
    \draw (v9) to (v10);
    \draw (v10) to (v11);

    \draw[very thick,dotted,transform canvas={xshift=-2pt}] (v2) to (v1);
    \draw[very thick,dotted] (v1) to[bend left] (v6);
    \draw[very thick,dotted,transform canvas={xshift=1.4pt,yshift=-1.4pt}] (v6) to (v10);

    \draw[very thick,dotted] (v5) to (v7);
    \draw[very thick,dotted,transform canvas={xshift=-1.4pt,yshift=1.4pt}] (v7) to (v11);

    \draw[thick,dashed] (v2) to (v5);
  \end{tikzpicture}
  \caption{Illustration of the proof of
    Theorem~\ref{thm:hard:response}. The two terminal sets in the
    instance of VDST are marked by red and blue vertices,
    respectively. Dotted lines represent follower's edges of cost $0$,
    whereas the dashed line represents the follower's edge in~$\bar Y$
    having cost $1$. One feasible solution to VDST is marked by red
    and blue edges; a corresponding leader's solution consists of the
    red, blue, and purple edges.}
  \label{fig:bmst-r}
\end{figure}
An illustration of this construction is given in
Fig.~\ref{fig:bmst-r}.  We now show that the answer to this instance
of BMST-R is yes if and only if the answer to the given VDST instance
is yes.

Assume that $T, T'\subseteq E$ are vertex-disjoint trees such that~$T$
spans~$S$ and~$T'$ spans~$S'$. Since~$G$ is connected, we may assume
that~$T\cup T'$ covers all vertices of~$G$, by connecting all
non-covered vertices to either $T$ or $T'$ arbitrarily. We claim that
the leader's choice~$X\coloneqq T\cup T'$ forces the follower to
respond with~$\bar Y$. Indeed, the follower's preferred edges $e$ with
$d(e) = 0$ would all produce cycles, while~$\{s_1,s'_1\}$ needs to be
added to turn~$X$ into a spanning tree.

Now assume that there exists a leader's solution~$X$ forcing the
follower to respond with exactly the set~$\bar Y$. Since the latter
prefers edges from~$E^f\setminus \bar Y$, the leader must prevent him
from adding any of those, i.e., all vertices in~$S$ are connected
by~$X$ and the same is true for the vertices in~$S'$. On the other
hand, since the follower chooses~$\{s_1,s'_1\}$, the sets~$S$ and~$S'$
cannot be connected by~$X$. Hence~$X$ contains two vertex-disjoint
trees spanning~$S$ and~$S'$, respectively.  
\end{proof}
Note that, similar to the proof of~Theorem~\ref{thm:hard}, there is
some freedom in the construction of the follower's edge set $E^f$ in
this proof; see Remark~\ref{rem:treestructure}. Instead of the paths
given by $\big\{\{s_i,s_{i+1}\}\mid i=1,\dots,r_1\big\}$ and
$\big\{\{s'_i,s'_{i+1}\}\mid i=1,\dots,r'-1\big\}$, one could choose
any other graph structure spanning the vertices in $S$ and $S'$,
respectively. Therefore, Theorem~\ref{thm:hard:response} does not only
hold for sets~$E^f$ forming a path, but also for many other
topologies.

Using the same construction as in Lemma~\ref{lem:EfE}, one can show
that the result of Theorem~\ref{thm:hard:response} holds for instances
in~$\IEfE$ as well. Moreover, since the leader's costs are not
relevant in the problem BMST-R, Theorem~\ref{thm:hard:response}
trivially remains true for any specific choice of leader's costs, in
particular in the case of uniform leader's costs.

\section{Fixed-parameter tractability}
\label{sec:fpt}

It is easy to see that BMST is tractable when the number of edges
controlled by the leader is bounded. In fact, we have
\begin{theorem} \label{thm:fpt:leaderedges}
  BMST is fixed-parameter tractable in the number of edges controlled
  by the leader.
\end{theorem}
\begin{proof}
  If $k=|E^\ell|$, the leader can choose between at most~$2^k$
  different solutions. Computing the follower's response and the
  corresponding objective function value is possible in polynomial
  time.
\end{proof}

We now turn to the question whether BMST is fixed-parameter tractable
in the number of edges controlled by the \emph{follower}, which is
much more involved. In fact, we are not able to answer it in
general. However, we will show some results related to this question.
We start by considering the problem BMST-R introduced in the previous
section.  In the proof of Theorem~\ref{thm:hard:response}, a
connection between BMST-R and VDST was established in order to prove
NP-completeness. It turns out that this relation is also useful for
translating positive results from VDST to BMST-R. More precisely, the
fact that VDST is fixed-parameter tractable in the total
number~$\sum_{i=1}^k |S_i|$ of
terminals~\cite{robertson1990,robertson1995} can be used to prove the
fixed-parameter tractability of BMST-R in terms of~$|E^f|$.
\begin{theorem} \label{thm:fpt:response}
  BMST-R is fixed-parameter tractable in the number of edges
  controlled by the follower.
\end{theorem}
\begin{proof}
  Consider an instance of BMST-R with graph $G = (V, E^\ell \cup
  E^f)$. Let~$V^f$ be the set of all end vertices of edges in
  $E^f$. The algorithm proceeds as follows: all partitions of~$V^f$
  into non-empty subsets are enumerated. For a given
  partition~$S_1,\dots,S_k$, the problem VDST on $(V, E^\ell)$ is
  solved by the algorithm given in~\cite{robertson1995}. Note that
  the graph $(V, E^\ell)$ does not have to be connected, but the definition of
  VDST and the algorithm can be used anyway. If the result is
  negative, the partition is discarded. Otherwise, let~$T_1,\dots,T_k$
  be a corresponding solution of VDST and extend the sets~$T_i$ such
  that~$X\coloneqq \bigcup_{i=1}^kT_i$ covers all vertices, while the $T_i$ must
  remain vertex-disjoint. This is possible, since we assume that~$G$
  is connected. Next, compute the follower's response~$Y'$ to~$X$. If
  it agrees with~$\bar Y$, stop and return ``yes'' and, if desired,
  the set~$X$. If the end of the enumeration is reached, return
  ``no''.

  The correctness of the algorithm immediately follows from the fact
  that the follower's response only depends on whether two vertices
  in~$V^f$ are connected by the leader or not, and all possible
  situations are enumerated. For the running time, note that the
  number and size of the enumerated partitions only depend on $|V^f|
  \leq 2 |E^f|$, but not on the size of the overall graph.
\end{proof}
The algorithm proposed in the proof of Theorem~\ref{thm:fpt:response}
can actually be used for enumerating all possible follower's
responses, along with one inducing leader's choice for each
response. Unfortunately, Theorem~\ref{thm:fpt:response} does not imply
that the problem of computing the \emph{best} leader's choice
enforcing a given response~$\bar Y$ is fixed-parameter tractable; see
the discussion below, so that we cannot derive that BMST itself is
fixed-parameter tractable in the number of edges controlled by the
follower. However, all leader's choices enforcing a given follower's
response consist of the same number of edges because every spanning
tree in the overall graph has the same number of edges. Hence, if the
leader has uniform costs on the edges in~$E^\ell$, this algorithm can
be used to solve BMST, leading to the following result:
\begin{corollary}\label{cor:fpt:uniform}
  BMST with $c(e) = \bar c$ for all $e \in E^\ell$, for some constant
  $\bar c \geq 0$, is fixed-parameter tractable in the number of edges
  controlled by the follower.
\end{corollary}
However, different costs on the edges in $E^\ell$ cannot be handled
easily. In particular, we cannot use the reduction in
Lemma~\ref{lem:uniform} to make the costs uniform, since it increases
the size of~$E^f$ by $\sum_{e \in E^\ell} (c(e) - 1)$. The result only
carries over to instances where the latter sum is bounded by some
function in the original number of follower's edges. Unfortunately, we
are not able to answer the question whether
Corollary~\ref{cor:fpt:uniform} also holds for arbitrary weights, but
we conjecture that this is not the case.  In fact, there is some
evidence that BMST is not easy to solve even for a \emph{fixed} number
of edges controlled by the follower. To justify this conjecture, we
will establish a relation between BMST and the optimization version of
VDST, the \emph{shortest vertex-disjoint Steiner trees} problem:

\medskip

(SVDST) Given a connected graph~$G=(V,E)$ with edge
lengths~$\ell\colon E\rightarrow\mathbb{R}_{\geq 0}$ and~$k$~disjoint
sets~$S_1,\dots,S_k\subseteq V$, find vertex-disjoint trees $T_1,
\dots, T_k \subseteq E$ such that $T_i$ spans $S_i$ for~$i=1,\dots,k$,
minimizing their total length $\sum_{i = 1}^k \ell(T_i)$, or decide
that such trees do not exist.

\medskip

Given that already the decision problem VDST is a very difficult
problem, it can be expected that SVDST is very hard as well. In fact,
even for the special case in which each set $S_i$ consists of only two
vertices, which is called the \emph{shortest vertex-disjoint paths}
(SVDP) problem, there are a lot of open complexity questions. Considerable
research has been devoted to SVDP for~$k = 2$. Very recently, a
randomized polynomial-time algorithm for this case has been
developed~\cite{bjoerklund2019}. To the best of our knowledge, no
deterministic polynomial-time algorithm for $k = 2$ nor the complexity
of SVDP for any fixed $k \geq 3$ is known.  According to the next
result, presenting an efficient algorithm for BMST with a fixed number
$|E^f| = 2k$ of edges controlled by the follower would settle these
open questions for $k$, and even similar ones about the more
general problem SVDST. In particular, an efficient algorithm
for BMST with~$|E^f|=4$ would lead to an efficient algorithm for SVDP
with~$k=2$.
\begin{theorem} \label{thm:hard2}
  SVDST with fixed number $K \coloneqq  \sum_{i = 1}^k |S_i|$ can be
  polynomially reduced to BMST with~$K$ edges controlled by the
  follower.
\end{theorem}
\begin{proof}
  Given an instance of SVDST as defined above, we construct an
  instance of BMST as follows. We extend~$G = (V, E)$ by one
  vertex~$s_0$, i.e., we set~$V' \coloneqq  V \cup \{s_0\}$. The edges
  controlled by the leader are given by $E^\ell \coloneqq  E \cup E^\ell_0$,
  where
  $$E^\ell_0 \coloneqq  \big\{\{s_0, v\} \mid v \in V \setminus
  \textstyle\bigcup_{i = 1}^k S_i\big\}\;.$$
  For the follower's edges, we introduce
  an arbitrary spanning tree on each vertex set~$S_i$ and call the set
  of these edges $E^f_0$. Moreover, for each~$i=1,\dots,k$, we select
  a vertex~$s_i \in S_i$ arbitrarily and introduce a follower's
  edge~$\{s_{i-1},s_i\}$. Together with~$E^f_0$, these edges form
  the set $E^f$. The cost function for the leader is defined as
  \[
  c(e) \coloneqq 
  \begin{cases}
    \ell(e) + M,	&	\text{if } e \in E ,\\
    M,	&	\text{if } e \in E^\ell_0 ,\\
    M |V|, &        \text{if } e \in E^f_0 ,\\
    0, &	\text{if } e \in E^f \setminus E^f_0,
  \end{cases}
  \] 
  where~$M\coloneqq \sum_{e\in E}\ell(e)+1$. The cost function for the follower is given by
  \[
  d(e) \coloneqq 
  \begin{cases}
    0, &        \text{if } e \in E^f_0 ,\\
    1, &	\text{if } e \in E^f \setminus E^f_0.
  \end{cases}
  \] 
  Clearly, this construction is polynomial, with~$|E^f|= k + \sum_{i =
    1}^k (|S_i| - 1) = K$;
  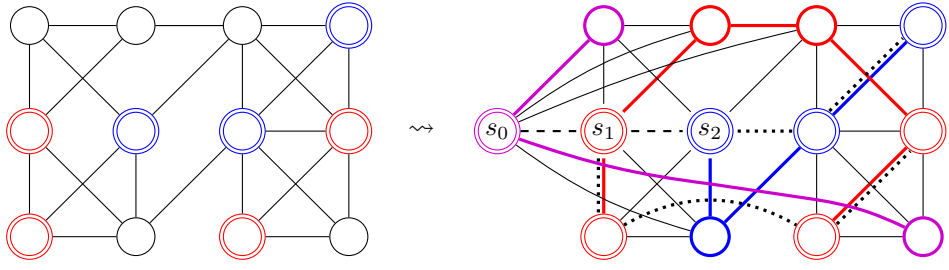
\begin{figure}
    \centering
    \small
    \begin{tikzpicture}[scale=0.7]
      \node[red,draw,circle,minimum width = 0.6cm] (v1) at (0,0) {};
      \node[red,draw,circle,minimum width = 0.5cm] at (0,0) {};
      \node[red,draw,circle,minimum width = 0.6cm] (v2) at (0,2) {};
      \node[red,draw,circle,minimum width = 0.5cm] at (0,2) {};
      \node[draw,circle,minimum width = 0.5cm] (v3) at (0,4) {};
      \node[draw,circle,minimum width = 0.5cm] (v4) at (2,0) {};
      \node[blue,draw,circle,minimum width = 0.6cm] (v5) at (2,2) {};
      \node[blue,draw,circle,minimum width = 0.5cm] at (2,2) {};
      \node[draw,circle,minimum width = 0.5cm] (v5b) at (2,4) {};
      \node[red,draw,circle,minimum width = 0.6cm] (v6) at (4,0) {};
      \node[red,draw,circle,minimum width = 0.5cm] at (4,0) {};
      \node[blue,draw,circle,minimum width = 0.6cm] (v7) at (4,2) {};
      \node[blue,draw,circle,minimum width = 0.5cm] at (4,2) {};
      \node[draw,circle,minimum width = 0.5cm] (v8) at (4,4) {};
      \node[draw,circle,minimum width = 0.5cm] (v9) at (6,0) {};
      \node[red,draw,circle,minimum width = 0.6cm] (v10) at (6,2) {};
      \node[red,draw,circle,minimum width = 0.5cm] at (6,2) {};
      \node[blue,draw,circle,minimum width = 0.6cm] (v11) at (6,4) {};
      \node[blue,draw,circle,minimum width = 0.5cm] at (6,4) {};

      \draw (v1) to (v2);
      \draw (v1) to (v4);
      \draw (v1) to (v5);
      \draw (v2) to (v3);
      \draw (v2) to (v4);
      \draw (v2) to (v5b);
      \draw (v3) to (v5);
      \draw (v3) to (v5b);
      \draw (v4) to (v5);
      \draw (v4) to (v7);
      \draw (v5) to (v8);
      \draw (v5b) to (v8);
      \draw (v5b) to (v8);
      \draw (v6) to (v7);
      \draw (v6) to (v9);
      \draw (v6) to (v10);
      \draw (v7) to (v8);
      \draw (v7) to (v9);
      \draw (v7) to (v10);
      \draw (v7) to (v11);
      \draw (v8) to (v10);
      \draw (v8) to (v11);
      \draw (v9) to (v10);
      \draw (v10) to (v11);
    \end{tikzpicture}
    \quad\raisebox{1.66cm}{$\leadsto$}~\quad
    \begin{tikzpicture}[scale=0.7]
      \definecolor{dpurple}{rgb}{0.8, 0.0, 0.8};

      \node[circle,minimum width = 0.5cm] at (-2,2) {$s_0$};
      \node[dpurple,draw,circle,minimum width = 0.6cm] (v0) at (-2,2) {};
      \node[dpurple,draw,circle,minimum width = 0.5cm] at (-2,2) {};

      \node[red,draw,circle,minimum width = 0.6cm] (v1) at (0,0) {};
      \node[red,draw,circle,minimum width = 0.5cm] at (0,0) {};
      \node[circle,minimum width = 0.5cm] (v2) at (0,2) {$s_1$};
      \node[red,draw,circle,minimum width = 0.6cm] at (0,2) {};
      \node[red,draw,circle,minimum width = 0.5cm] at (0,2) {};
      \node[dpurple,draw,circle,minimum width = 0.5cm,very thick] (v3) at (0,4) {};
      \node[blue,draw,circle,minimum width = 0.5cm,very thick] (v4) at (2,0) {};
      \node[circle,minimum width = 0.5cm] (v5) at (2,2) {$s_2$};
      \node[blue,draw,circle,minimum width = 0.6cm] at (2,2) {};
      \node[blue,draw,circle,minimum width = 0.5cm] at (2,2) {};
      \node[red,draw,circle,minimum width = 0.5cm,very thick] (v5b) at (2,4) {};
      \node[red,draw,circle,minimum width = 0.6cm] (v6) at (4,0) {};
      \node[red,draw,circle,minimum width = 0.5cm] at (4,0) {};
      \node[blue,draw,circle,minimum width = 0.6cm] (v7) at (4,2) {};
      \node[blue,draw,circle,minimum width = 0.5cm] at (4,2) {};
      \node[red,draw,circle,minimum width = 0.5cm,very thick] (v8) at (4,4) {};
      \node[dpurple,draw,circle,minimum width = 0.5cm,very thick] (v9) at (6,0) {};
      \node[red,draw,circle,minimum width = 0.6cm] (v10) at (6,2) {};
      \node[red,draw,circle,minimum width = 0.5cm] at (6,2) {};
      \node[blue,draw,circle,minimum width = 0.6cm] (v11) at (6,4) {};
      \node[blue,draw,circle,minimum width = 0.5cm] at (6,4) {};

      \draw[red,very thick] (v1) to (v2);
      \draw (v1) to (v4);
      \draw (v1) to (v5);
      \draw (v2) to (v3);
      \draw (v2) to (v4);
      \draw[red,very thick] (v2) to (v5b);
      \draw (v3) to (v5);
      \draw (v3) to (v5b);
      \draw[blue,very thick] (v4) to (v5);
      \draw[blue,very thick] (v4) to (v7);
      \draw (v5) to (v8);
      \draw[red,very thick] (v5b) to (v8);
      \draw (v6) to (v7);
      \draw (v6) to (v9);
      \draw[red,very thick] (v6) to (v10);
      \draw (v7) to (v8);
      \draw (v7) to (v9);
      \draw (v7) to (v10);
      \draw[blue,very thick] (v7) to (v11);
      \draw[red,very thick] (v8) to (v10);
      \draw (v8) to (v11);
      \draw (v9) to (v10);
      \draw (v10) to (v11);

      \draw[dpurple,very thick] (v0) to (v3);
      \draw (v0) to[bend right=10] (v4);
      \draw (v0) to[bend left=10] (v5b);
      \draw (v0) to[bend left=5]  (v8);
      \node (x1) at (0,1) {};
      \draw[dpurple,very thick] (v0) to[bend right=6] (2.0,1.0) to (4.0,0.7) to[bend left=8] (v9);
      
      \draw[very thick,dotted,transform canvas={xshift=-2pt}] (v2) to (v1);
      \draw[very thick,dotted] (v1) to[bend left] (v6);
      \draw[very thick,dotted,transform canvas={xshift=1.4pt,yshift=-1.4pt}] (v6) to (v10);

      \draw[very thick,dotted] (v5) to (v7);
      \draw[very thick,dotted,transform canvas={xshift=-1.4pt,yshift=1.4pt}] (v7) to (v11);

      \draw[thick,dashed] (v0) to (v2);
      \draw[thick,dashed] (v2) to (v5);
    \end{tikzpicture}
    \caption{Illustration of the proof of Theorem~\ref{thm:hard2} for
      $k=2$ and $K=7$. The two terminal sets~$S_1$ and~$S_2$ in the
      instance of SVDST are marked by red and blue vertices,
      respectively. Dotted lines represent follower's edges of cost $0$,
      whereas the dashed lines represent the follower's edges having follower's
      cost $1$. An optimum leader's solution is given by the colored
      edges, where red and blue edges correspond to an optimum solution
      of the original instance of SVDST and purple edges connect all
      vertices not covered by the latter with the auxiliary
      vertex~$s_0$.}
    \label{fig:red_svdst}
  \end{figure}
  an illustration is given in
  Fig.~\ref{fig:red_svdst}. We claim that the given instance of SVDST
  is feasible if and only if the optimum value of the constructed
  BMST instance is smaller than~$M(|V|-k+1)$, and that in this case
  the optimum values differ by exactly~$M(|V|-k)$.

  So first assume that vertex-disjoint trees $T_i$ spanning $S_i$ for
  $i=1,\dots,k$, exist. Then consider the leader's choice~$X$
  consisting of all edges contained in any of the trees~$T_i$ and, for
  each vertex~$v\in V$ not belonging to any tree, the
  edge~$\{s_0,v\}$. We have $|X| = |V| - k$ because $X$ forms a forest
  with $k+1$ connected components on~$|V|+1$ vertices. The follower's
  response to~$X$ is~$Y \coloneqq  \big\{\{s_{i-1},s_i\} \mid
  i=1,\dots,k\big\}$ with $c(Y) = 0$.  In summary, the objective
  value of~$X$ is
  $$c(X) + c(Y) = \sum_{i=1}^k\ell(T_i)+M(|V|-k) + 0 < M(|V|-k+1)\;.$$
  
  For the other direction, consider any feasible leader's choice~$X$
  in the constructed instance of BMST and assume that it has an
  objective value less than~$M(|V|-k+1)$. Then for all $i=1,\dots,k$,
  all vertices in $S_i$ must be connected in~$X$, as otherwise the
  follower would choose an edge with leader's cost~$M|V|\ge
  M(|V|-k+1)$. Moreover, since each leader's edge costs at least~$M$
  and the final tree must have~$|V|$ edges, the only way to achieve a
  weight less than~$M(|V|-k+1)$ is to take exactly $|V| - k$ edges and
  make the follower choose all edges~$\{s_{i-1},s_i\}$
  for~$i=1,\dots,k$. It follows that~$X$ has~$k+1$ components
  containing exactly one of the vertices~$s_0,\dots,s_k$
  each. Thus,~$X$ contains disjoint trees~$T_i$ spanning $S_i$ with total
  weight
  $$\sum_{i=1}^k \ell(T_i) = \sum_{e\in X}(c(e)-M)=\sum_{e\in X}c(e)-M(|V|-k)\;.$$
  This concludes the proof.
\end{proof}
As in Theorem~\ref{thm:hard} and
Theorem~\ref{thm:hard:response}, also other topologies of the
follower's edges are possible; see Remark~\ref{rem:treestructure}.  We
emphasize that Theorem~\ref{thm:hard2} gives a second proof for the
NP-hardness of BMST, if we do not bound the number of edges controlled
by the follower. In particular, it shows that BMST is at least as hard
as SVDST. However, the reduction used in the proof of
Theorem~\ref{thm:hard2} is not approximation-preserving, so that the
negative result of Theorem~\ref{thm:hard} concerning approximability
does not follow from Theorem~\ref{thm:hard2}.

While Theorem~\ref{thm:hard2} makes it unlikely that BMST is
fixed-parameter tractable in the number of follower's edges, we will
show next that at least approximating BMST within a factor of $2$ is
fixed-parameter tractable in the same parameter. As a first step, we
show that a similar result holds for a variant of SF defined as
follows:

\medskip

(SF+) Given a connected graph $G = (V, E)$ with edge
lengths~$\ell\colon E\rightarrow \mathbb{R}_{\geq 0}$ and~$k$~disjoint
sets $S_1, \dots, S_k \subseteq V$, find a forest $F\subseteq E$ of
minimum total length~$\ell(F)$, such that each terminal set $S_i$ is
connected in the graph $(V, F)$ and every vertex in~$V \setminus
\bigcup_{i = 1}^k S_i$ is connected to one of the sets $S_i$.

\medskip

The difference from the usual Steiner forest problem is hence that
in addition to connecting each terminal set~$S_i$, all non-terminals
need to be connected to one of the terminal sets.

\begin{theorem} \label{thm:fpt:2apx:SF+}
  The problem of approximating SF+ within a factor of 2 is fixed-parameter
  tractable in the total number of terminals.
\end{theorem}
\begin{proof}
  We use the fact that the Steiner forest problem is
  fixed-parameter tractable, which can be seen as follows: for the
  classical Steiner tree problem, an exact algorithm with running time
  $O(3^{|S|}|V|)$, where~$S$ is the set of terminals, is
  well-known~\cite{dreyfus71}. This can be extended to the Steiner
  forest problem in the following way: enumerate all partitions of the
  set $\{S_1, \dots, S_k\}$ of terminal sets, each resulting in a
  coarser partition~$S_1', \dots, S_r'$ of the set $\bigcup_{i = 1}^k
  S_i$ of all terminals. Now solve the Steiner tree problem for each
  terminal set $S_i'$ and merge the resulting~$r$ edge sets in order to obtain a
  feasible solution of the Steiner forest problem. Obviously, the best
  solution obtained in this way is optimum.
  
  Now we compute a solution to the problem SF+ in the following way:
  first, compute an optimum solution $F$ of the corresponding Steiner
  forest problem, for example using the algorithm described
  above. Second, merge the set $\bigcup_{i = 1}^k S_i$ of all
  terminals, together with all non-terminals that are connected to a
  terminal by edges in $F$, into a single new vertex. Now compute a
  minimum spanning tree~$T$ in the resulting graph and return the set
  $F \cup T$ as a solution to the given instance of SF+.

  Clearly, the solution is feasible for SF+ and the running time is
  the same as the one of the applied Steiner forest algorithm because
  the running time for the computation of a minimum spanning tree is
  negligible. It remains to show that it is a 2-approximation. For
  this, observe that both $F$ and $T$ have at most the cost of an
  optimum solution to SF+, since every such solution must contain a
  Steiner forest having at least the cost of $F$, as well as a
  spanning tree in the graph in which $T$ is a minimum spanning tree.
\end{proof}
Theorem~\ref{thm:fpt:2apx:SF+} now allows us to show the desired result about BMST:
\begin{theorem}
\label{thm:fpt:2apx}
The problem of approximating BMST within a factor of 2 is fixed-para\-meter
tractable in the number of edges controlled by the follower.
\end{theorem}
\begin{proof}
  We may assume that~$(V,E^\ell)$ is connected, since the construction
  according to Lemma~\ref{lem:Elconn} does not increase the number of
  follower's edges. The idea of the algorithm is to enumerate all
  possible follower's solutions and apply
  Theorem~\ref{thm:fpt:2apx:SF+} to each of them. More precisely, as
  in the proof of Theorem~\ref{thm:fpt:response}, we consider the
  set~$V^f \subseteq V$ of vertices incident to a follower's edge and
  enumerate all possible partitions of $V^f$ into non-empty disjoint
  sets. For a fixed partition $S_1, \dots, S_k$, we solve SF+ on the
  leader's graph~$(V, E^\ell)$ with edge lengths defined by the
  leader's cost function~$c$, using the algorithm given in
  Theorem~\ref{thm:fpt:2apx:SF+}, and obtain some forest~$X\subseteq
  E^\ell$. Next, we compute the follower's response to $X$ and, if
  there is a feasible response $Y$, store it together with $X$ as a
  candidate for our final solution. Finally, we return the candidate
  solution minimizing the total weight $c(X) + c(Y)$.

  Clearly, the running time of this algorithm is as desired. Moreover,
  it computes a feasible solution to BMST if there is one.  It remains
  to prove that in this case the algorithm always computes a
  2-approximate solution. For this, let $X^*$ be an optimum leader's
  solution of the given BMST instance, together with the follower's
  response~$Y^*$, and let $S^*_1, \dots, S^*_k$ be the partition of
  $V^f$ corresponding to the connected components of $(V, X^*)$. Then
  $X^*$ is a (not necessarily optimum) solution for SF+ corresponding
  to this partition.  Let $X$ be the solution for SF+ computed by the
  algorithm presented above when considering this partition. Since we
  use a 2-approximation algorithm for SF+, we have $c(X) \leq
  2c(X^*)$. Moreover, the partition of~$V^f$ induced by~$X$ is
  either~$S^*_1, \dots, S^*_k$ or a coarser one, which implies that
  the follower's response~$Y$ to~$X$ is a subset of~$Y^*$ by
  Lemma~\ref{lem:Efforest}, based on the deterministic behavior of the
  follower.  Altogether, we now obtain
  $$c(X) + c(Y)
  \leq 2c(X^*)+c(Y^*)\leq 2(c(X^*) + c(Y^*))\;,$$
  which shows the desired result.
\end{proof}

\section{Approximation algorithm for BMST}
\label{sec:algo}

The previous section showed that already questions about
fixed-parameter tractability of BMST and related problems can be hard
to answer. In this section, we present a polynomial-time
$(|V|-1)$-approximation algorithm for BMST.
\begin{theorem}\label{thm:apx}
  BMST admits a polynomial-time $(|V|-1)$-approximation algorithm.
\end{theorem}
\begin{proof}
  The algorithm starts with an empty leader's solution $X \coloneqq
  \emptyset$ and iteratively adds leader's edges to $X$. At the same
  time the graph~$G=(V,E^\ell\cup E^f)$, initially given as part of
  the considered BMST instance, is modified in each iteration of the
  algorithm. More specifically, in each iteration, we first apply
  Corollary~\ref{cor:Efforest} in order to turn~$E^f$ into a
  forest. Then, in the current graph $G=(V, E^\ell \cup E^f)$, we
  compute a minimum spanning tree~$T$ according to the leader's cost
  function~$c$. Let~$T^\ell \coloneqq T \cap E^\ell$ be the part of
  the spanning tree that is controlled by the leader, and add the
  edges in~$T^\ell$ to~$X$. If $T^\ell = T$ or $T^\ell = \emptyset$,
  we stop and output $X$ as the leader's solution. Otherwise, we
  contract the edges in $T^\ell$ and start the next iteration.

  The algorithm clearly runs in polynomial time, since we perform at
  most~$|V|-1$ iterations, and in each iteration we apply the
  polynomial reduction of Corollary~\ref{cor:Efforest} and compute a
  minimum spanning tree. It is also not hard to see that the
  algorithm computes a feasible solution: if it stops with $T^\ell =
  T$, the leader's solution~$X$ already forms a spanning tree in the
  original graph. Otherwise, it stops with $T^\ell = \emptyset$. In
  this case, the follower is able to complete $X$ to a spanning tree,
  for example using the edges in~$T$.  It remains to show that the
  objective value of $X$ is at most $|V|-1$ times the optimum value.

  We prove this by induction on the number~$|V|$ of vertices. If
  $|V|=2$, the statement is clearly true since we may assume that we
  only have two edges, one leader's and one follower's edge.  So let
  us assume that for some arbitrary but fixed $n \in \mathbb{N}$ the
  statement is true for all graphs that have at most $n$ vertices.
  Let $I$ be an instance with $G=(V, E^\ell \cup E^f)$, where $E^f$ is
  a forest and $|V| = n+1$. Let $T$ be the spanning tree that is
  computed in the first iteration of the algorithm. If it stops after
  the first iteration, i.e., if $T^\ell = T$ or $T^\ell = \emptyset$,
  either the leader or the follower chooses the whole tree $T$, while
  the other player chooses $\emptyset$; note that
  for~$T^\ell=\emptyset$ the follower can only choose~$T$ as response,
  as~$E^f$ is cycle-free. In both cases, the leader's objective value
  is~$c(T)$, which is clearly optimum. Otherwise, we have $c(T^\ell)
  \leq c(T) \leq OPT(I)$, where $OPT(I)$ denotes the value of an
  optimum solution to instance $I$.  Let~$\hat{I}$ with the graph
  $\hat{G} = (\hat{V}, \hat{E})$ be the instance that is considered in
  the second iteration, i.e., after contracting $T^\ell$.  Observe
  that by Lemma~\ref{lem:Efforest}, we have $OPT(\hat{I}) \leq
  OPT(I)$, since~$\hat{I}$ arises from~$I$ by contracting certain
  edges of the graph. Furthermore, a solution $\hat{X}$ to~$\hat{I}$
  with follower's response $\hat{Y}$ can be augmented to a solution to
  $I$ by simply adding~$T^\ell$, such that $\hat{Y}$ remains the
  follower's response.  Finally, observe that~$| \hat{V} | \leq n$ and
  hence the induction hypothesis holds, i.e., the solution $\hat{X}$
  to $\hat{I}$ produced by the algorithm is
  a~$(|\hat{V}|-1)$-approximation, where $|\hat{V}| - 1 \leq |V| - 2$.
  Putting things together, we derive that
  \begin{eqnarray*}
    c(X)+c(\hat{Y}) & = & c(T^\ell) + c(\hat{X}) + c(\hat{Y})\\
    & \leq & OPT(I) + (|\hat{V}| - 1) OPT(\hat{I})\\
    & \leq & (|V| - 1) OPT(I)
  \end{eqnarray*}
  holds for the objective value of the leader's solution $X = T^\ell
  \cup \hat{X}$ returned by the algorithm. 
\end{proof}

\section{Bottleneck objective}
\label{sec:bottleneck}

In this section, we consider variants of BMST in which one or both of
the two decision makers have a bottleneck objective function instead
of a sum objective, i.e., they pay only for the most expensive edge in
their solution. Recall that when the follower has a bottleneck
objective, we have to distinguish two variants of this objective,
namely minimizing either~$\max_{e \in Y}d(e)$ or~$\max_{e \in X\cup
  Y}d(e)$, i.e., the follower either takes only his own edges into
account or both the edges chosen by the leader and by himself. As
already mentioned in the introduction, these variants are not
equivalent, in contrast to the corresponding variants in the sum
objective case. The problem version in which the follower considers
only his own edges can be seen as a special case of the one in which
he considers all edges by setting $d(e) \coloneqq 0$ for all $e \in
E^\ell$.

Consider the example depicted in Fig.~\ref{fig:example} and assume
that the leader still has a sum objective, but the follower has a
bottleneck objective. In his response to the leader's choice shown in
Fig.~\ref{fig:example}, the follower could now also choose the edge
$\{v_3, v_5\}$ instead of the edge $\{v_3, v_6\}$. Both options
are optimum from the follower's perspective. Under the optimistic
assumption, the follower would choose $\{v_3, v_6\}$ because it is
better for the leader. But under the pessimistic assumption, the
follower would choose $\{v_3, v_5\}$ instead, increasing the leader's
objective value by $4$.

Shi \emph{et al.}~\cite{shi19} showed that BMST is tractable as soon
as the leader or the follower (or both) optimize a bottleneck
objective. However, the general assumption in~\cite{shi19} is that the
follower's and the leader's edge sets are not disjoint, but that the
follower controls all edges, or, equivalently, that instances belong
to~$\IEfE$. Note that, in the definition of $\IEfE$, we have to
require the parallel edges to have not only the same leader's, but
also the same follower's cost now. Without this assumption, the
tractability results do not hold anymore in general. In fact, we will
see that most cases are NP-hard then.  Gassner~\cite{gassner02}
developed two polynomial-time algorithms without the assumption that
the follower controls all edges, namely for the cases in which the
leader has a bottleneck objective and the follower either has a sum
objective or a bottleneck objective, restricting to the pessimistic
problem version in the latter case. In this case, however, she always
assumes the follower to minimize~$\max_{e \in Y} d(e)$. We generalize
this result to the case of the follower's objective being~$\max_{e \in
  X \cup Y} d(e)$ and slightly simplify her other algorithm.
Moreover, our hardness results show that these are the only two cases
which are polynomial-time solvable in general, unless~P$\,=\,$NP. An
overview of the different cases and results is given in
Table~\ref{tab:bottleneck:results}.
\begin{table}
  \small
  \centering
  \begin{tabular}{|c|c|c|l|}
    \hline
    Leader & Follower & Assumption & Results \\
    \hline
    \hline
    S & S & opt/pess & NP-hard (Theorem~\ref{thm:hard}) \\
    \hline
    S & BN & pess & P for $\IEfE$ and $\max_{e \in X \cup Y} d(e)$~(\cite{shi19}) \\
    &&& NP-hard (Corollary~\ref{cor:hard:sum:bottleneck:pess}) \\
    \hline
    S & BN & opt & P for $\IEfE$ and $\max_{e \in X \cup Y} d(e)$~(\cite{shi19}) \\
    &&& NP-hard (Theorem~\ref{thm:hard:sum:bottleneck}) \\
    \hline
    BN & S & opt/pess & P for $\IEfE$~(\cite{shi19}) \\
    &&& P~(\cite{gassner02} and Theorem~\ref{thm:easy:bottleneck:sum}) \\
    \hline
    BN & BN & pess & P for $\IEfE$ and $\max_{e \in X \cup Y} d(e)$~(\cite{shi19}) \\
    &&& P for $\max_{e \in Y} d(e)$~(\cite{gassner02}) \\
    &&& P (Theorem~\ref{thm:easy:bottleneck:bottleneck:pess}) \\
    \hline
    BN & BN & opt & P for $\IEfE$ and $\max_{e \in X \cup Y} d(e)$~(\cite{shi19}) \\
    &&& NP-hard (Theorem~\ref{thm:hard:bottleneck:bottleneck:opt}) \\
    \hline
  \end{tabular}
  \caption{Results for all variants of BMST with sum (S) or
    bottleneck (BN) objective functions, assuming an optimistic (opt)
    or pessimistic (pess) setting.}
  \label{tab:bottleneck:results}
\end{table}

\begin{theorem}
  \label{thm:easy:bottleneck:sum}
  The variant of BMST in which the leader has a bottleneck objective and
  the follower has a sum objective can be solved in polynomial time.
\end{theorem}

\begin{proof}
  We first present the algorithm: for each~$\gamma\in C\coloneqq
  \{c(e)\mid e\in E^\ell\} \cup \{0\}$, the leader considers the set
  $E_\gamma\coloneqq \{e\in E^\ell\mid c(e)\le\gamma\}$ and chooses
  any edge set~$X_\gamma \subseteq E_\gamma$ consisting of a spanning
  tree in each connected component of~$G_\gamma \coloneqq
  (V,E_\gamma)$. Let $Y_\gamma$ be the corresponding response of the
  follower and $c_\gamma$ the resulting leader's objective value,
  where~$c_\gamma\coloneqq \infty$ in case the follower cannot
  extend~$X_\gamma$ to a spanning tree. Finally,
  choose~$\gamma^*\in\argmin_{\gamma\in C}c_\gamma$ and
  return~$X_{\gamma^*}$ as optimum solution.
  
  The algorithm clearly runs in polynomial time, so it remains to show
  that~$X_{\gamma^*}$ is indeed an optimum solution. For this, it
  suffices to show that, for any~$\gamma$, choosing a
  solution~$X\subseteq E_\gamma$ with the same bottleneck cost cannot
  yield a smaller objective function value than~$c_\gamma$.  Since the
  objective function of the follower is a sum,
  Lemma~\ref{lem:Efforest} applies, thus his response~$Y$ to $X$ is a
  superset of~$Y_\gamma$. Now the cost of~$X\cup Y$ (in the leader's
  bottleneck objective) is at least the cost of~$X_\gamma\cup
  Y_\gamma$.  \end{proof}

We now turn to the case in which both leader and follower have a
bottleneck objective. Then, the above algorithm does not work in
general because Lemma~\ref{lem:Efforest} is not true in case the
follower has a bottleneck objective. However, Gassner~\cite{gassner02}
showed that the same algorithm solves the problem version in which both
leader and follower have a bottleneck objective, the pessimistic setting
is assumed and the follower's objective is to minimize $\max_{e \in Y}
d(e)$. We next prove that a generalized form of
the algorithm can be used to solve the problem version with follower's
objective $\max_{e \in X \cup Y} d(e)$, completing the investigation
of all polynomial-time solvable cases.
\begin{theorem}
  \label{thm:easy:bottleneck:bottleneck:pess}
  The variant of BMST in which both leader and follower have a
  bottleneck objective and the pessimistic setting is assumed can be
  solved in polynomial time.
\end{theorem}
\begin{proof}
  The algorithm works as follows: for all~$e_c,e_d \in E^\ell$ such
  that~$c(e_c)\ge c(e_d)$ and~$d(e_c)\leq d(e_d)$, and such that
  either $e_c=e_d$ or $e_c$ and $e_d$ are not parallel, consider the
  set
  $$E_{e_c, e_d} \coloneqq  \{e \in E^\ell \mid c(e) \leq c(e_c) \text{ and }
  d(e) \leq d(e_d)\}\;.$$ The leader chooses any edge set~$X_{e_c,
    e_d} \subseteq E_{e_c, e_d}$ with $e_c, e_d \in X_{e_c, e_d}$ that
  consists of a spanning tree in each connected component of~$G_{e_c,
    e_d} \coloneqq (V,E_{e_c, e_d})$. Let $Y_{e_c, e_d}$ be the
  corresponding response of the follower and $c_{e_c, e_d}$ the
  resulting objective value, where~$c_{e_c, e_d}\coloneqq \infty$ in
  case the follower cannot extend~$X_{e_c, e_d}$ to a spanning
  tree. If the case $c(e_c) = d(e_d) = 0$ does not occur, consider~$X_0
  \coloneqq \emptyset$ as an additional candidate. Finally,
  choose~$(e_c^*, e_d^*)$ minimizing $c_{e_c, e_d}$ and
  return~$X_{e_c^*, e_d^*}$.

  The algorithm clearly runs in polynomial time, so it remains to show
  that~$X_{e_c^*, e_d^*}$ is indeed an optimum solution. For this, let
  $e_c, e_d \in E^\ell$ be such that $c(e_c)$ and~$d(e_d)$ are the
  maximum leader's and follower's edge costs, respectively, among a
  leader's optimum solution $X \subseteq E^\ell$, assuming $X \neq
  \emptyset$. We show that $X$ cannot have a smaller objective
  function value than~$c_{e_c, e_d}$.

  If $X$ is a maximal forest in $G_{e_c, e_d}$, it leads to the same
  follower's response and hence the same objective function value as
  $X_{e_c, e_d}$. Otherwise, we may assume that~$X \subset X_{e_c,
    e_d}$. The follower cannot achieve a better objective value when
  responding to $X$ than to $X_{e_c, e_d}$. Hence, by the pessimistic
  assumption, the maximum leader's edge cost among the follower's
  response cannot be smaller in the former than in the latter
  case. Since the maximum leader's and follower's edge costs
  among~$X$ and $X_{e_c, e_d}$, respectively, are the same, it follows
  that $X$ cannot lead to a smaller objective function value than
  $X_{e_c, e_d}$.
\end{proof}
Turning to the hardness results, we will reuse several ideas from
Section~\ref{sec:complexity} and Section~\ref{sec:fpt} that can be
applied directly or need to be changed slightly for the bottleneck
cases. First, note that in Theorem~\ref{thm:hard}, the follower's sum
objective can be easily replaced by a bottleneck objective, assuming
the pessimistic setting:
\begin{corollary} \label{cor:hard:sum:bottleneck:pess}
  The variant of BMST in which the leader has a sum objective, the
  follower has a bottleneck objective and the pessimistic setting is
  assumed, cannot be approximated to within a factor of
  $\frac{96}{95}$ in polynomial time, unless~P$\,=\,$NP, even if $E^f$
  is a tree.
\end{corollary}
\begin{proof}
  We can use the same reduction from the Steiner forest problem as in
  the proof of Theorem~\ref{thm:hard}. For sake of simplicity, the
  follower's cost function can now be defined as $d(e) \coloneqq 0$
  for all edges $e \in E \cup E^f$. Then the follower's objective
  value is always $0$ and his decision is only guided by the
  pessimism.  For this definition of $d$, both variants of the
  follower's bottleneck objective function, $\max_{e \in Y} d(e)$ and
  $\max_{e \in X\cup Y} d(e)$, are equivalent, hence this proof
  clearly holds for both of them.  The pessimistic assumption about
  the follower's behavior here is equivalent to the behavior of a
  follower having a cost function of $d(e) \coloneqq -c(e)$ for all $e
  \in E^f$ and a sum objective, which is equivalent to the problem
  variant which is reduced to in Theorem~\ref{thm:hard}; see also
  Remark~\ref{rem:mmst}.
\end{proof}

Corollary~\ref{cor:hard:sum:bottleneck:pess} cannot be
easily adapted to the optimistic assumption. However, the hardness of
this case can be concluded using Theorem~\ref{thm:hard:response}:
\begin{theorem}
  \label{thm:hard:sum:bottleneck}
  All variants of BMST where the leader has a sum objective and the
  follower has a bottleneck objective are NP-hard, even if~$c(e)=1$
  for all~$e\in E^\ell$.
\end{theorem}
\begin{proof}
  First, note that the proof of Theorem~\ref{thm:hard:response} works
  without modification if the follower has a bottleneck objective, for
  both the optimistic and pessimistic setting, no matter if the
  follower is taking only his own or all edges into account; for the
  latter case, define~$d(e) \coloneqq 0$ for all $e \in
  E^\ell$. Hence, all corresponding modifications of the problem
  BMST-R are NP-complete as well, even if $|\bar Y| = 1$ and $E^f$
  forms a path on a subset of the vertex set.

  We now show that BMST-R can be reduced to BMST in all these
  variants, assuming a sum objective for the leader, which proves the
  desired result.
  Consider an instance of BMST-R, consisting of a graph $G = (V,
  E^\ell \cup E^f)$, a follower's objective~$d$ and a set $\bar Y
  \subseteq E^f$. Define a leader's cost function $c\colon E \rightarrow
  \mathbb{R}_{\geq 0}$ by setting    
  \[
  c(e) \coloneqq 
  \begin{cases}
    1,	&	\text{if } e \in E^\ell ,\\
    0,	&	\text{if } e \in \bar Y ,\\
    |V|, &	\text{if } e \in E^f \setminus \bar Y.
  \end{cases}
  \] 
  We claim that the answer to the given instance of BMST-R is yes
  if and only if the leader's optimum solution value in this BMST
  instance is at most $|V| - |\bar Y| - 1$.

  Assume that $X \subseteq E^\ell$ is a leader's solution such that
  $\bar Y$ is the follower's response to~$X$. Choosing $X$ then yields
  a leader's objective value of $c(X) + c(\bar Y) = |V| - |\bar Y| -
  1$, since $X$ and $\bar Y$ form a tree and hence together have $|V|
  - 1$ edges.  Conversely, assume that the leader can achieve an
  objective value of at most $|V| - |\bar Y| - 1$. By construction,
  this is only possible if the follower's response is exactly~$\bar Y$
  and the leader thus chooses~$|V| - |\bar Y| - 1$ of her
  edges. Hence, the follower's response $\bar Y$ can be enforced.
\end{proof}
The proof of Theorem~\ref{thm:hard:sum:bottleneck} does not carry over
to cases in which the leader has a bottleneck objective, because the
reduction from BMST-R to BMST does not work there. However, the case
in which both leader and follower have a bottleneck objective,
assuming the optimistic setting, is NP-hard as well, which can be shown
using similar ideas as in the proof of
Theorem~\ref{thm:hard:sum:bottleneck}.
\begin{theorem}
\label{thm:hard:bottleneck:bottleneck:opt}
  The variant of BMST in which both the leader and the follower have a
  bottleneck objective and the optimistic setting is assumed is NP-hard.
\end{theorem}
\begin{proof}
  We show the result by reduction from VDST, restricted
  to~$k=2$. Given an instance of VDST consisting of a connected graph
  $G = (V, E)$ and disjoint vertex sets $S = \{s_1, \dots, s_r\}$ and
  $S' = \{s_1', \dots, s_{r'}'\}$, we define an instance of BMST by
  adding a vertex $s_0$ to $V$, setting $E^\ell \coloneqq E$,
  $$E^f_0\coloneqq \big\{\{s_i,s_{i+1}\}\mid i=1,\dots,r-1\big\}
  \cup\big\{\{s'_i,s'_{i+1}\}\mid i=1,\dots,r'-1\big\}\;,$$ and $E^f \coloneqq  E^f_0
  \cup\big\{\{s_1,s'_1\}, \{s_0, s_1\}, \{s_0, s_1'\}\big\}$,
  where the leader's and follower's costs are defined as follows:
  \begin{align*}
  c(e) &\coloneqq 
  \begin{cases}
    0, & \text{if } e \in E^\ell \cup \big\{ \{s_1, s_1'\}, \{s_0, s_1\} \big\} \\
    1, & \text{if } e \in E^f_0 \cup \big\{ \{s_0, s_1'\} \big\}
  \end{cases}\\
  d(e) &\coloneqq 
  \begin{cases}
    0, & \text{if } e \in E^\ell \cup E^f_0 \cup \big\{ \{s_0, s_1'\} \big\} \\
    1, & \text{if } e \in \big\{ \{s_1, s_1'\}, \{s_0, s_1\} \big\}
  \end{cases}  
  \end{align*}
  This construction is illustrated
  in Fig.~\ref{fig:red2_svdst}.  
  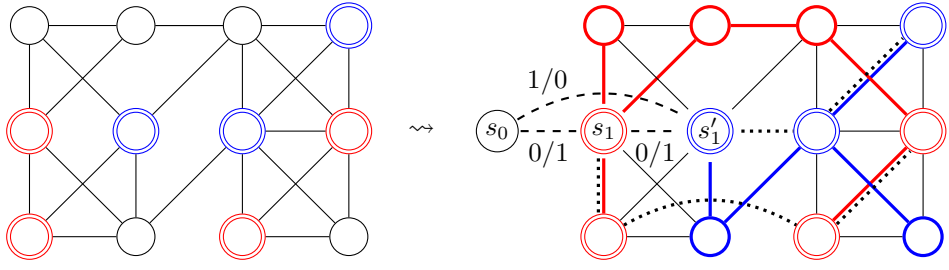
\begin{figure}
    \centering
    \small
    \begin{tikzpicture}[scale=0.7]
      \node[red,draw,circle,minimum width = 0.6cm] (v1) at (0,0) {};
      \node[red,draw,circle,minimum width = 0.5cm] at (0,0) {};
      \node[red,draw,circle,minimum width = 0.6cm] (v2) at (0,2) {};
      \node[red,draw,circle,minimum width = 0.5cm] at (0,2) {};
      \node[draw,circle,minimum width = 0.5cm] (v3) at (0,4) {};
      \node[draw,circle,minimum width = 0.5cm] (v4) at (2,0) {};
      \node[blue,draw,circle,minimum width = 0.6cm] (v5) at (2,2) {};
      \node[blue,draw,circle,minimum width = 0.5cm] at (2,2) {};
      \node[draw,circle,minimum width = 0.5cm] (v5b) at (2,4) {};
      \node[red,draw,circle,minimum width = 0.6cm] (v6) at (4,0) {};
      \node[red,draw,circle,minimum width = 0.5cm] at (4,0) {};
      \node[blue,draw,circle,minimum width = 0.6cm] (v7) at (4,2) {};
      \node[blue,draw,circle,minimum width = 0.5cm] at (4,2) {};
      \node[draw,circle,minimum width = 0.5cm] (v8) at (4,4) {};
      \node[draw,circle,minimum width = 0.5cm] (v9) at (6,0) {};
      \node[red,draw,circle,minimum width = 0.6cm] (v10) at (6,2) {};
      \node[red,draw,circle,minimum width = 0.5cm] at (6,2) {};
      \node[blue,draw,circle,minimum width = 0.6cm] (v11) at (6,4) {};
      \node[blue,draw,circle,minimum width = 0.5cm] at (6,4) {};
      
      \draw (v1) to (v2);
      \draw (v1) to (v4);
      \draw (v1) to (v5);
      \draw (v2) to (v3);
      \draw (v2) to (v4);
      \draw (v2) to (v5b);
      \draw (v3) to (v5);
      \draw (v3) to (v5b);
      \draw (v4) to (v5);
      \draw (v4) to (v7);
      \draw (v5) to (v8);
      \draw (v5b) to (v8);
      \draw (v5b) to (v8);
      \draw (v6) to (v7);
      \draw (v6) to (v9);
      \draw (v6) to (v10);
      \draw (v7) to (v8);
      \draw (v7) to (v9);
      \draw (v7) to (v10);
      \draw (v7) to (v11);
      \draw (v8) to (v10);
      \draw (v8) to (v11);
      \draw (v9) to (v10);
      \draw (v10) to (v11);
    \end{tikzpicture}
    \quad\raisebox{1.66cm}{$\leadsto$}~\quad
    \begin{tikzpicture}[scale=0.7]
      \definecolor{dpurple}{rgb}{0.8, 0.0, 0.8};
      
      \node[circle,minimum width = 0.5cm] at (-2,2) {$s_0$};
      \node[draw,circle,minimum width = 0.55cm] at (-2,2) {};
      
      \node[red,draw,circle,minimum width = 0.6cm] (v1) at (0,0) {};
      \node[red,draw,circle,minimum width = 0.5cm] at (0,0) {};
      \node[circle,minimum width = 0.5cm] (v2) at (0,2) {$s_1$};
      \node[red,draw,circle,minimum width = 0.6cm] at (0,2) {};
      \node[red,draw,circle,minimum width = 0.5cm] at (0,2) {};
      \node[red,draw,circle,minimum width = 0.5cm,very thick] (v3) at (0,4) {};
      \node[blue,draw,circle,minimum width = 0.5cm,very thick] (v4) at (2,0) {};
      \node[circle,minimum width = 0.5cm] (v5) at (2,2) {$s_1'$};
      \node[blue,draw,circle,minimum width = 0.6cm] at (2,2) {};
      \node[blue,draw,circle,minimum width = 0.5cm] at (2,2) {};
      \node[red,draw,circle,minimum width = 0.5cm,very thick] (v5b) at (2,4) {};
      \node[red,draw,circle,minimum width = 0.6cm] (v6) at (4,0) {};
      \node[red,draw,circle,minimum width = 0.5cm] at (4,0) {};
      \node[blue,draw,circle,minimum width = 0.6cm] (v7) at (4,2) {};
      \node[blue,draw,circle,minimum width = 0.5cm] at (4,2) {};
      \node[red,draw,circle,minimum width = 0.5cm,very thick] (v8) at (4,4) {};
      \node[blue,draw,circle,minimum width = 0.5cm,very thick] (v9) at (6,0) {};
      \node[red,draw,circle,minimum width = 0.6cm] (v10) at (6,2) {};
      \node[red,draw,circle,minimum width = 0.5cm] at (6,2) {};
      \node[blue,draw,circle,minimum width = 0.6cm] (v11) at (6,4) {};
      \node[blue,draw,circle,minimum width = 0.5cm] at (6,4) {};
      
      \draw[red,very thick] (v1) to (v2);
      \draw (v1) to (v4);
      \draw (v1) to (v5);
      \draw[red,very thick] (v2) to (v3);
      \draw (v2) to (v4);
      \draw[red,very thick] (v2) to (v5b);
      \draw (v3) to (v5);
      \draw (v3) to (v5b);
      \draw[blue,very thick] (v4) to (v5);
      \draw[blue,very thick] (v4) to (v7);
      \draw (v5) to (v8);
      \draw[red,very thick] (v5b) to (v8);
      \draw (v6) to (v7);
      \draw (v6) to (v9);
      \draw[red,very thick] (v6) to (v10);
      \draw (v7) to (v8);
      \draw[blue,very thick] (v7) to (v9);
      \draw (v7) to (v10);
      \draw[blue,very thick] (v7) to (v11);
      \draw[red,very thick] (v8) to (v10);
      \draw (v8) to (v11);
      \draw (v9) to (v10);
      \draw (v10) to (v11);

      \draw[thick,dashed] (v0) to node[below] {$0 / 1$} (v2);
      \draw[thick,dashed] (v0) to[bend left] node[pos=0.18,above] {$1 / 0$} (v5);
      \draw[thick,dashed] (v2) to node[below] {$0 / 1$} (v5);
      
      \draw[very thick,dotted,transform canvas={xshift=-2pt}] (v2) to (v1);
      \draw[very thick,dotted] (v1) to[bend left] (v6);
      \draw[very thick,dotted,transform canvas={xshift=1.4pt,yshift=-1.4pt}] (v6) to (v10);

      \draw[very thick,dotted] (v5) to (v7);
      \draw[very thick,dotted,transform canvas={xshift=-1.4pt,yshift=1.4pt}] (v7) to (v11);
    \end{tikzpicture}
    \caption{Illustration of the proof of
      Theorem~\ref{thm:hard:bottleneck:bottleneck:opt}. The two terminal
      sets~$S$ and~$S'$ in the instance of VDST are marked by red and
      blue vertices, respectively. Dotted lines represent edges in~$E^f_0$
      having costs~$1 / 0$, dashed lines have costs as specified. Solid
      edges are controlled by the leader and have costs~$0 / 0$.}
    \label{fig:red2_svdst}
  \end{figure}
  We now show that the answer to the given instance of VDST is yes if and
  only if the leader's optimum value in the constructed instance of BMST is
  $0$. Since~$d(e)=0$ for all~$e\in E^\ell$, the following arguments
  hold for both types of follower's objective functions, $\max_{e \in
    Y} d(e)$ as well as $\max_{e \in X \cup Y} d(e)$.

  Assume that $T, T' \subseteq E$ are vertex-disjoint trees such that
  $T$ spans~$S$ and $T'$ spans~$S'$. Since $G$ is connected, we may
  assume that $T \cup T'$ covers all vertices of~$G$, by connecting
  all non-covered vertices to either $T$ or $T'$ arbitrarily. If the
  leader chooses $X \coloneqq T \cup T'$ as her solution, the follower
  must take any two of the three edges $\{s_1, s_1'\}$, $\{s_0, s_1\}$
  and $\{s_0, s_1'\}$ in order to complete $X$ to a spanning tree. As
  the follower's objective value is $1$ for any of these choices and
  we assume the optimistic setting, his response is $Y \coloneqq
  \big\{ \{s_1, s_1'\}, \{s_0, s_1\} \big\}$, resulting in a leader's
  objective value of $0$.

  For the other direction, assume that the leader can achieve an
  objective value of $0$. This means that the follower uses the edge
  $\{s_0, s_1\}$ in order to connect the vertex $s_0$ to the original
  graph. Since this edge is more expensive than $\{s_0, s_1'\}$ for
  the follower, he will only do that if he is also forced to connect
  the vertices $s_1$ and~$s_1'$, because otherwise, he can always
  achieve an objective value of $0$. Hence, the leader must not
  connect $s_1$ and~$s_1'$. Moreover, all vertices in $S$ have to be
  connected by the leader, as well as all vertices in $S'$, in order
  to prevent the follower from taking any edge from $E^f_0$. Thus, the
  leader's solution contains two vertex-disjoint trees spanning~$S$
  and $S'$, respectively.
\end{proof}

The proof of Theorem~\ref{thm:hard:bottleneck:bottleneck:opt} shows
that even computing any approximate solution is NP-hard, because the
reduction only relies on distinguishing whether the optimum value is
$0$ or $1$.

\section{Conclusion}
\label{sec:conc}

In this paper, we investigated the computational complexity of the
bilevel minimum spanning tree problem.  After giving some structural
insights about the problem, we proved that BMST is NP-hard, thus
answering a conjecture stated by Shi \emph{et al.}~\cite{shi19}.
Furthermore, we considered the parameterized complexity of the problem
in the number of edges controlled by the follower and showed that the
problem is at least as hard as the shortest vertex-disjoint Steiner
trees problem, parameterized by the number of terminal vertices,
giving some evidence that the problem might be intractable even for a
fixed number of follower's edges.  Finally, we considered several
variants of BMST in which at least one of the decision makers has a
bottleneck objective function and gave a complete complexity
classification of all these variants.

It is still open whether BMST is solvable for a fixed number of
follower's edges or even fixed-parameter tractable in this
parameter. Also given the close relation to the shortest
vertex-disjoint paths problem, we consider this to be an interesting
open question.  Moreover, the approximability of BMST is an
interesting question to study further, given that the best
approximation ratio achieved is $|V| - 1$.

As a generalization of BMST, one could consider the bilevel minimum
matroid basis~(BMMB) problem, in which both decision makers together
have to compute a basis of a given matroid.  We think that some of our
structural results can be generalized to the matroid setting.  On the
one hand, it would be interesting to see which of the positive results
can be generalized to BMMB. Furthermore, we are curious if the
negative results could be strengthened, in particular if the follower
only controls a fixed number of elements.

\bibliographystyle{abbrv}
\bibliography{literature}
  
\end{document}